\documentclass[nolineno,a4paper,UKenglish,thm-restate]{socg-lipics-v2021}

\hideLIPIcs  

\graphicspath{{./Figures/}}
\Copyright{Todor Anti\'c \and Anna Margarethe Limbach  \and Niloufar Fuladi \and Pavel Valtr}


\usepackage{tikz}
\usetikzlibrary{calc}
\usepackage{amsfonts}
\usepackage{amsmath}
\usepackage{amssymb}
\usepackage{amsthm}
\usepackage{thmtools}
\usepackage{graphicx} 
\usepackage{url}
\usepackage{tcolorbox} 
\usepackage{hyperref,cleveref} 

\newcommand{\commentOut}[1]{}

\title{Hypercube drawings with no long plane paths}

\author{Todor {Anti\'c}}{Department of Applied Mathematics, Faculty of Mathematics and Physics, Charles University, Czech Republic}{todor@kam.mff.cuni.cz}{https://orcid.org/0009-0008-6521-7987}{supported by grant no. 23-04949X of the Czech Science Foundation (GA\v{C}R) and by SVV–2023–260699.}

\author{Niloufar Fuladi}{Computer Science Institute, Faculty of Mathematics and Physics, Charles University, Czech Republic}{nfuladi@iuuk.mff.cuni.cz}{https://orcid.org/0009-0004-5870-2504}{}

\author{Anna Margarethe Limbach}{Computer Science Institute, Faculty of Mathematics and Physics, Charles University, Czech Republic}{limbach@iuuk.mff.cuni.cz}{https://orcid.org/0000-0002-3606-153X}{}

\author{Pavel Valtr}{Department of Applied Mathematics, Faculty of Mathematics and Physics, Charles University, Czech Republic}{valtr@kam.mff.cuni.cz}{https://orcid.org/0000-0001-9267-3605}{supported by grant no. 23-04949X of the Czech Science Foundation (GA\v{C}R).}

\acknowledgements{We would like to thank organizers and participants of KAMAK 2025 workshop where this research was initiated for creating a pleasant working atmosphere. }

\authorrunning{T. Anti\'c, N. Fuladi,  A. M. Limbach, P. Valtr}

\usepackage{apptools}
\newcommand{\restateref}[1]{\IfAppendix{\hyperref[#1]{$\star$}}{\hyperref[#1*]{$\star$}}}


\begin{document}

\maketitle

\begin{abstract}
We study the existence of plane substructures in drawings of the $d$-dimensional hypercube graph $Q_d$. 
We construct drawings of $Q_d$ which contain no plane subgraph with more than $2d-2$ edges, no plane path with more than $2d-3$ edges, and no plane matching of size more than $2d-4$. On the other hand, we prove that every rectilinear drawing of $Q_d$ with vertices in convex position contains a plane path of length $d$ (if $d$ is odd) or $d-1$ (if $d$ is even). We also prove that if a graph $G$ is a plane subgraph of every drawing of $Q_d$ for a sufficiently large $d$, then $G$ is necessarily a forest of caterpillars.
Lastly, we give a short proof of a generalization of a result by Alpert et al. [Cong. Numerantium, 2009] on the maximum rectilinear crossing number of $Q_d$.
\end{abstract}

\section{Introduction}

Among the first graphs introduced in any undergrad graph theory course are complete graphs, complete bipartite graphs and $d$-dimensional hypercube graphs (which we refer to simply as $d$-cubes). It is therefore not surprising that drawings of complete~\cite{Aichholzer2024,completerectcrossing,Gioan2022,Kynl2009,brego2013,Arroyo2021,Arroyo2021a,Balko2025} and complete bipartite~\cite{CardinalFelsner2018,deKlerk2014,Aichholzer2023, Lenhart2023} graphs  are a frequent topic of research in graph drawing, with many results and interesting problems. However, one might find surprising that drawings of $d$-cubes have not gained such interest. As far as we are aware, except for various versions of crossing number problems~\cite{cuberectcrossing,Kainen1972,Skora1992,MR1775301,Shahrokhi1997,Faria2008,kainen2025,Hammack2023}, drawings of $d$-cubes have not been studied, yet. 

The $d$-cube $Q_d$ is a graph with $2^d$ vertices labeled by binary strings of length $d$ in which two vertices are connected by an edge if their labels differ exactly in one digit. A \emph{rectilinear drawing} is a drawing of a graph with vertices represented by points in the plane and edges represented by line segments between the vertices. We say that a drawing is \emph{plane} if no two edges cross. We call a drawing \emph{convex-geometric} if it is rectilinear and its vertices are in convex position. We call a drawing \emph{simple} if edges are represented by simple curves with no self intersections, and every two curves intersect in at most one point which is either a common vertex or a proper crossing in the interior of both edges (no tangential ``touches''). 

We investigate the problem of finding plane subgraphs in (rectilinear) drawings of $Q_d$. This problem has been considered for complete and complete bipartite graphs extensively. If we are trying to find a single plane Hamiltonian path (or cycle), the problem is almost trivial for rectilinear drawings of $K_n$. However, if we relax the setting slightly and consider simple drawings, the problem becomes much harder. In fact, the existence of a plane Hamiltonian path in every simple drawing of $K_n$ was conjectured in 1988 by Rafla~\cite{Alsolami2012Auth}, and even after almost $40$ years of effort by multiple authors~\cite{Aichholzer2024,Bergold2025,Pach2005,Suk2023,Fulek2014,Garca2023,GarcaOlaverri2021,Pach2003,keszegh}, the best known result is that every simple drawing of $K_n$ contains a plane path of length $O(\log(n)/\log\log(n))$~\cite{Aichholzer2024}. We mention that finding large plane matchings, cycles, and trees has also been considered in the literature~\cite{Pach2003,Garca2023,Aichholzer2024}. In the case of rectilinear drawings, since finding one plane Hamiltonian path is easy, finding multiple such paths has been considered~\cite{Kindermann2023,arxiv}.

Finding plane spanning paths in rectilinear drawings of complete bipartite graphs has proven to be an interesting problem (without considering more relaxed drawings). It is known that there are drawings with no plane path of length more than $(4-2\sqrt{2})n + o(n)$~\cite{Cska2022}, where $n$ is the total number of vertices. On the other hand, it is known that a plane path of length $\frac{n}{2}$ always exists~\cite{soukup2024}. Since plane Hamiltonian paths do not exist in every drawing of the complete bipartite graph, several authors tried to give sufficient and necessary conditions for a drawing to contain a plane Hamiltonian path~\cite{Cibulka2009,mulzer2020,Cska2022,Ricci2025,Kynl2008,soukup2024}.  

\subsection{Our results}

We start by considering plane subgraphs of drawings of hypercubes in the most restrictive setting, where we only consider convex-geometric drawings. As our first result, we show that every convex-geometric drawing of $Q_d$ contains a plane path of \emph{length} $d$ (that is, a path with $d$ edges). Actually, this is an easy consequence of a famous theorem of Perles, mentioned in~\cite{perles}. We include a proof of this theorem for the purpose of self-containment.

\begin{restatable}{theorem}{longplanepathsCor}
\label{cor_longplanepaths}
For every $d\ge 2$, any convex-geometric drawing of $Q_d$ contains a plane path of length at least $d$ if $d$ is odd, and $d-1$ if $d$ is even.
\end{restatable}
In Theorem~\ref{thm:drawingswithnolongpaths} below, we show that Theorem~\ref{cor_longplanepaths} is tight up to a multiplicative constant. Theorem~\ref{thm:drawingswithnolongpaths} follows from the following more general result (up to additive constant $1$).

\begin{restatable}{theorem}{planesubgraphrotateddrawing}
\label{thm-planesubgraphrotateddrawing}
For every $d \ge 3$, there is a convex-geometric drawing of $Q_d$ which does not contain a plane subgraph with more than $2d-2$ edges.
\end{restatable}

\begin{restatable}{theorem}{drawingswithnolongpaths}
\label{thm:drawingswithnolongpaths}
For every $d \ge 4$, there is a convex-geometric drawing of $Q_d$ which does not contain a plane path longer than $2d-3$.
\end{restatable}


\begin{restatable}{theorem}{drawingswithnolargematchings}
\label{thm:drawingswithnolargematchings}
For every $d \ge 4$, there is a convex-geometric drawing of $Q_d$ which does not contain a plane matching on more than $2d-4$ edges.
\end{restatable}

Note that Theorem~\ref{thm:drawingswithnolargematchings} is again tight up to a multiplicative constant, since by Theorem~\ref{cor_longplanepaths}, we know that every convex-geometric drawing of $Q_d$ contains a plane path of length $d$ or $d-1$ and, therefore, a plane matching on $\lfloor\frac{d}{2}\rfloor$ or $\lfloor \frac{d-1}{2} \rfloor$ edges, respectively. 

As our last result regarding plane subgraphs of convex-geometric drawings of $Q_d$, we provide a necessary condition for a graph $G$ to be a plane subgraph of every rectilinear drawing of $Q_d$ (for sufficiently large $d$). 

\begin{restatable}{theorem}{graphsembeddableinQd}
\label{thm:graphsembeddableinQd}
If there is a $d \ge 2$ such that every rectilinear drawing of $Q_d$ contains a fixed graph $G$ as a plane subgraph, then $G$ is a forest of caterpillars.
\end{restatable}

We then shift our focus away from plane subgraphs of drawings of the $d$-cube and use a slight modification of our construction to revisit the problem of computing the \emph{maximum rectilinear crossing number} of a $d$-cube. 
The maximum rectilinear crossing number of a graph $G$ (denoted by CR$_{\max}(G)$) is the maximum number of crossings, taken over all rectilinear drawings of $G$. This parameter was studied for many graph classes, such as cycles~\cite{cyclesrectcross}, complete graphs~\cite{completerectcrossing} and $k$-regular graphs~\cite{kregrectcross}, among others. Alpert et al.~\cite{cuberectcrossing}~investigated this question for $Q_d$. They proved that CR$_{\max}(Q_d) \ge 2^{d-2}(2^{d-1}(d^2-2d+3)-d^2-1)$ and conjectured that this is tight. We give a significantly shorter proof of the following more general result 
(terms ``length-regular'' and ``length profile'' are defined in Subsection~\ref{subs:prelim}).

\begin{restatable}{theorem}{lengthregcrossings}
\label{thm:lengthregcrossings}
Let $d\ge 2$, and let $G$ be a $d$-regular graph. Then every length-regular drawing of $G$ with length profile $(\ell_1, \ell_2, \ldots, \ell_d)$ has exactly  $\frac{|V(G)|}{2} \sum_{i=1}^d (\ell_i - 1) \left(i - \frac{1}{2}\right)$ crossings.
\end{restatable}

Since the drawings that we construct are length-regular, Theorem~\ref{thm:lengthregcrossings} can be applied on them.  It gives the lower bound on CR$_{\max}(Q_d)$ mentioned above by a simple calculation. 
Interestingly, even though our construction is different from the one presented by Alpert et al., in the resulting drawings, the vertices lie in the same cyclic order on the convex hull and they are \emph{weakly isomorphic}, i.e., the same pairs of edges cross.

Lastly, we consider drawings of $Q_d$ that are not necessarily convex-geometric and prove the following. 

\begin{restatable}{proposition}{fourpathrectilinear}
\label{prop:4pathrectilinear}
Any rectilinear drawing of $Q_d$, for $d \ge 3$, contains a plane path of length at least $4$.
\end{restatable}

In the case of simple drawings of $Q_d$, we prove that Proposition~\ref{prop:4pathrectilinear} does not admit a straightforward generalization. 

\begin{restatable}{proposition}{QthreeNoPlanePath}
\label{prop:Q3NoPlanePath}
There exists a drawing of $Q_3$ with no plane path of length $4$.
\end{restatable}

Our paper is organized as follows. In the next section, we give some technical definitions that we use throughout the paper. In Section~\ref{sec:planesubgraphs}, we derive Theorem~\ref{cor_longplanepaths} and we present our first construction, which we use to prove Theorems~\ref{thm-planesubgraphrotateddrawing},~\ref{thm:drawingswithnolongpaths},~\ref{thm:drawingswithnolargematchings} and~\ref{thm:graphsembeddableinQd}. In Section~\ref{sec:maxcrossing}, we prove Theorem~\ref{thm:lengthregcrossings} and present our second construction which we use to give the above lower bound on CR$_{\max}(Q_d)$. In Section~\ref{sec:nonconvexdrawings}, we prove Propositions~\ref{prop:4pathrectilinear} and~\ref{prop:Q3NoPlanePath}. 
In the last section, we mention some open problems.



\subsection{Preliminaries}\label{subs:prelim}

Let $\mathcal{G}$ be a convex-geometric drawing of a graph $G$. The \emph{combinatorial length}, or \emph{length} for short, of an edge $e=uv$ in $\mathcal{G}$ is the graph distance between $u$ and $v$ in the cycle bounding the convex hull of the vertices of $\mathcal{G}$ (the edges of this cycle are not necessarily edges of the drawing); see Figure~\ref{fig:length}. Assume that $G=Q_d$ for some $d$. The possible lengths of the edges are $1,2,\dots, 2^{d-1}$. To a vertex $v\in \mathcal{G}$ we associate a tuple $\hat{v} = (\ell_1,\ell_2,\dots, \ell_d)$ where $\ell_1\ge \ell_2 \ge \cdots \ge \ell_d$ are the lengths of the edges adjacent to $v$. If $\hat{v}=\hat{u}$ for every two vertices $u,v \in \mathcal{G}$, we say that $\mathcal{G}$ is \emph{length-regular} and we call $\hat{v}$ the \emph{length profile} of $\mathcal{G}$. 

For a convex-geometric drawing $\mathcal{G}$ and two vertices $v,w\in \mathcal{G}$ we write $[v,w]$ for the set of vertices of $\mathcal{G}$ that are visited when traversing the convex hull of the vertices of $\mathcal{G}$ clockwise from $v$ to $w$; we call this an \emph{interval}. We write $(v,w)$ for $[v,w] \setminus \{v,w\}$. As mentioned before, the $d$-cube $Q_d$ is a graph whose vertices are binary strings of length $d$. We denote the vertices of $Q_d$ by $\overline{x}$, where $x$ is a string in $\{0,1\}^d$. Throughout the paper, we slightly abuse notation and identify vertices of a graph $G$ with their representation in a drawing of $G$. 

\begin{figure}[ht]
    \centering
    \begin{tikzpicture}[line cap=round,line join=round, scale=1.7]
\tikzset{
  vertex/.style={circle, fill, inner sep=1.2pt}
}

\foreach \i in {1,...,16} {
  \node[vertex] (v\i) at ({360/16*(\i-1)}:1cm) {};
   \fill (v\i) circle (1pt);
}
\node (c1) at ({360/16*(5-1)}:1.25cm) {};

\node  at ({360/16*(5-1)}:1.15cm) {$u$};
\node  at ({360/16*(12-1)}:1.15cm) {$v$};
\node  at (-.3,0) {$e$};

\draw[dotted]
\foreach \v [remember=\v as \prev (initially v16)] in
{v1,v2,v3,v4,v5,v6,v7,v8,v9,v10,v11,v12,v13,v14,v15,v16}
{
  (\prev) -- (\v)
};
\draw[thick, black, line width=0.8pt] (v5) -- (v12);
\draw (c1) arc[start angle=90, end angle=247.5, radius=1.25cm];
\end{tikzpicture}
    \caption{The length of $e$ is $7$: the number of convex hull edges within the depicted arc.}
    \label{fig:length}
\end{figure}
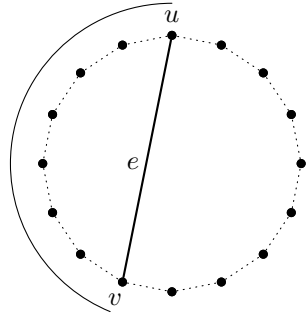

\section{Plane subgraphs in convex-geometric drawings of the \texorpdfstring{$\mathbf{d}$}{d}-cube}\label{sec:planesubgraphs}

We start by giving a lower bound on the length of a plane path that is present in any convex-geometric drawing of the $d$-cube.
The following result is a variation of a theorem by Perles, mentioned in~\cite[p. 292]{perles}, we prove it in the appendix.

\begin{restatable}[\restateref{lemma_longplanepaths}]{lemma}{longplanepaths}
\label{longplanepaths}
\label{lemma_longplanepaths}
Let $k \in \mathbb{N}$ and $G = (E, V)$ be a graph with $|V| = n$ and $|E| > kn$. Then, any convex-geometric drawing of $G$ contains a plane path of length at least $2k+1$.
\end{restatable}

\longplanepathsCor*

\begin{proof}
       $Q_d$ has $\frac{d}{2}n$ edges.
    If $d$ is odd, let $k:=\frac{d-1}{2}$. There is a plane path of length at least $2\frac{d-1}{2}+1=d$ by Lemma~\ref{lemma_longplanepaths}.
     If $d$ is even, let $k:=\frac{d-2}{2}$. There is a plane path of length $2\frac{d-2}{2}+1=d-1$ by Lemma~\ref{lemma_longplanepaths}.
\end{proof}

We now define a particular convex-geometric drawing of $Q_d$, which we use to find upper bounds for the maximum sizes of various plane subgraphs in convex-geometric drawings of $Q_d$. 
In our construction, the vertices of $Q_d$ are placed on a circle $C$ obtained from the interval $[0,1]\subset \mathbb{R}$ by identifying the endpoints. With slight abuse of notation, we identify the points of $C$ with the points in $[0,1)$. Now, we define a drawing $\mathcal{H}_2$ of $Q_2$ by mapping the vertices $\overline{00} \to 0$, $
\overline{01}\to\frac{1}{4}$, $\overline{11}\to \frac{1}{2}$, $\overline{10}\to \frac{3}{4}$. For $d\ge 3$, we construct the drawing $\mathcal{H}_d$ by taking two copies $\mathcal{H}', \mathcal{H}''$,  of $\mathcal{H}_{d-1}$ and rotating $\mathcal{H}''$ so that the vertex $\overline{00\dots0}\in \mathcal{H}''$ is placed on the point $\frac{1}{2}-\frac{1}{2^d}\in C$, in other words, we rotate the drawing by slightly less than $180^{\circ}$. We then connect every vertex in $\mathcal{H}'$ to its copy in $\mathcal{H}''$. Now, we map the vertices of $Q_d$ to $\mathcal{H}_d$ by mapping a vertex $\overline{0x}$ to the copy of $x$ in $\mathcal{H}'$ and $\overline{1x}$ to the copy of $x$ in $\mathcal{H}''$. See Figure~\ref{fig:H_d} for the construction of $\mathcal{H}_d$.

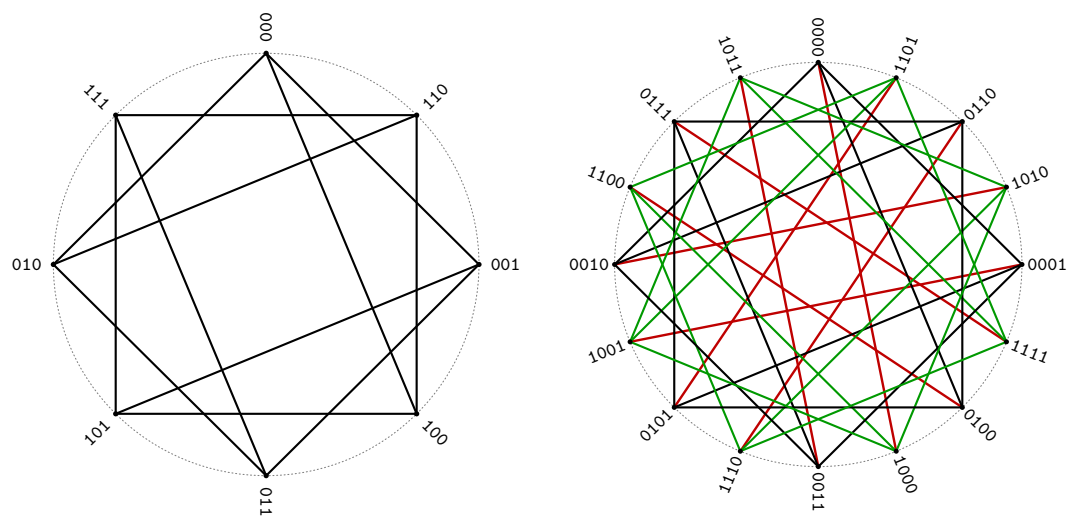
\begin{figure}[ht]

\begin{minipage}{0.4\textwidth}
    \begin{tikzpicture}[line cap=round,line join=round, scale=0.7]

\def\R{4}      
\def\Rl{4.5}   
\def\vdot{1.4pt}

\draw[densely dotted,gray] (0,0) circle (\R);

\foreach \k/\lab in {0/000,2/110,4/001,6/100,8/011,10/101,12/010,14/111}{
  \pgfmathsetmacro{\ang}{90 - 360*\k/16}
  \coordinate (v\lab) at (\ang:\R);
  \coordinate (l\lab) at (\ang:\Rl);
}

\tikzset{
  e1/.style={black, line width=0.9pt},
  e2/.style={black, line width=0.8pt},
  e3/.style={black, line width=0.8pt},
  e4/.style={black, line width=0.8pt},
}


\draw[e2] (v000) -- (v100);
\draw[e2] (v001) -- (v101);
\draw[e2] (v010) -- (v110);
\draw[e2] (v011) -- (v111);

\draw[e3] (v000) -- (v010);
\draw[e3] (v001) -- (v011);
\draw[e3] (v100) -- (v110);
\draw[e3] (v101) -- (v111);

\draw[e4] (v000) -- (v001);
\draw[e4] (v010) -- (v011);
\draw[e4] (v100) -- (v101);
\draw[e4] (v110) -- (v111);

\foreach \k/\lab in {0/000,2/110,4/001,6/100,8/011,10/101,12/010,14/111}{

  \path (v\lab) -- (l\lab)
    node[pos=1, sloped, font=\scriptsize\ttfamily, inner sep=1pt] {\lab};

  \fill (v\lab) circle (\vdot);
}
\end{tikzpicture}
\end{minipage}
\hspace{1.5cm}
\begin{minipage}{0.4\textwidth}
\begin{tikzpicture}[line cap=round,line join=round, scale=0.67]

\def\R{4}      
\def\Rl{4.5}   
\def\vdot{1.4pt}

\draw[densely dotted,gray] (0,0) circle (\R);

\foreach \k/\lab in {0/0000,1/1101,2/0110,3/1010,4/0001,5/1111,6/0100,7/1000,8/0011,9/1110,10/0101,11/1001,12/0010,13/1100,14/0111,15/1011}{
  \pgfmathsetmacro{\ang}{90 - 360*\k/16}
  \coordinate (v\lab) at (\ang:\R);
  \coordinate (l\lab) at (\ang:\Rl);
}

\tikzset{
  e1/.style={red!75!black, line width=0.9pt},
  e2/.style={black, line width=0.8pt},
  e3/.style={green!60!black, line width=0.8pt},
  e4/.style={green!60!black, line width=0.8pt},
}

\draw[e1] (v0000) -- (v1000);
\draw[e1] (v0001) -- (v1001);
\draw[e1] (v0010) -- (v1010);
\draw[e1] (v0011) -- (v1011);
\draw[e1] (v0100) -- (v1100);
\draw[e1] (v0101) -- (v1101);
\draw[e1] (v0110) -- (v1110);
\draw[e1] (v0111) -- (v1111);

\draw[e2] (v0000) -- (v0100);
\draw[e2] (v0001) -- (v0101);
\draw[e2] (v0010) -- (v0110);
\draw[e2] (v0011) -- (v0111);
\draw[e3] (v1000) -- (v1100);
\draw[e3] (v1001) -- (v1101);
\draw[e3] (v1010) -- (v1110);
\draw[e3] (v1011) -- (v1111);

\draw[e2] (v0000) -- (v0010);
\draw[e2] (v0001) -- (v0011);
\draw[e2] (v0100) -- (v0110);
\draw[e2] (v0101) -- (v0111);
\draw[e3] (v1000) -- (v1010);
\draw[e3] (v1001) -- (v1011);
\draw[e3] (v1100) -- (v1110);
\draw[e3] (v1101) -- (v1111);

\draw[e2] (v0000) -- (v0001);
\draw[e2] (v0010) -- (v0011);
\draw[e2] (v0100) -- (v0101);
\draw[e2] (v0110) -- (v0111);
\draw[e3] (v1000) -- (v1001);
\draw[e3] (v1010) -- (v1011);
\draw[e3] (v1100) -- (v1101);
\draw[e3] (v1110) -- (v1111);

\foreach \k/\lab in {0/0000,1/1101,2/0110,3/1010,4/0001,5/1111,6/0100,7/1000,8/0011,9/1110,10/0101,11/1001,12/0010,13/1100,14/0111,15/1011}{

  \path (v\lab) -- (l\lab)
    node[pos=1, sloped, font=\scriptsize\ttfamily, inner sep=1pt] {\lab};

  \fill (v\lab) circle (\vdot);
}

\end{tikzpicture}
\end{minipage}
\caption{The construction of $\mathcal{H}_4$ (right) from $\mathcal{H}_3$ (left). In the right picture, the edges in $\mathcal{H}'$ and $\mathcal{H}''$ are depicted in black and green, respectively, and the red edges depict the ones between these two copies of $\mathcal{H}_3$.}
\label{fig:H_d}
\end{figure}

Observe that, by definition, $\mathcal{H}_d$ is a length-regular drawing with length profile
\begin{equation*}
    (\ell_1,\ell_2,\dots, \ell_d)=(2^{d-1}-1, 2^{d-1}-2,2^{d-2}-4,\dots, 2^{d-1}-2^{d-2},2^{d-1}-{2^{d-2}}).
\end{equation*} For a vertex $v\in \mathcal{H}_d$, let $e_1^v,e_2^v,\dots,e_{d-2}^v$ be the edges adjacent to $v$ with length strictly larger than $2^{d-1}-2^{d-2}$, sorted by length (with $e_1^v$ being the longest) and let $p_v$ be the line through $v$ and the center of $C$. We associate with each such edge $e_i^v$ a \emph{direction} $d(e_i^v)= \pm1$ where $d(e_i^v)= +1 (-1)$ means that the other endpoint of $e_i^v$ lies to the left (right) of $p_v$ when viewed from $v$.  We call the tuple $(d(e_1^v),d(e_2^v),\dots,d(e_{d-2}^v))$ the \emph{length-rotation} of $v$. Note that there are $2^{d-2}$ length-rotations that can appear in $\mathcal{H}_d$, we start by observing that all of them indeed appear in $\mathcal{H}_d$. 

\begin{restatable}{mylemma}{allrotationshere}
\label{lem-allrotationshere}\label{cor:all-length-rotation}
For $d \ge 2$ and $\overline{v} \in \{+1, -1\}^{d-2}$, there are exactly two pairs of antipodal vertices in $\mathcal{H}_d$ with length-rotation $\overline{v}$. Thus, every possible length-rotation exists $4$ times.
\end{restatable}

\begin{proof}
    We proceed by induction on $d$. For $d=2$, the statement is clearly true. Assume that the statement holds for $d\ge 2$. Let $\overline{v}\in \{+1,-1\}^{d-2}$. By our inductive assumption, there are $4$ vertices in $\mathcal{H}_d$ with length-rotation $\overline{v}$, denote them by $u_1,u_2,u_3,u_4$. Recall that $\mathcal{H}_{d+1}$ is created from two copies $\mathcal{H},\mathcal{H}'$ of $\mathcal{H}_d$. Assume that $u_1,u_2,u_3,u_4\in \mathcal{H}$ and let $u_1',u_2',u_3',u_4'\in \mathcal{H}'$ be their copies. Now, note that each $u_i'$ has length-rotation $\overline{v}$ in $\mathcal{H}'$. In $\mathcal{H}_{d+1}$ we connect $u_i$ to $u_i'$ by an edge of length $2^{d-1}-1$ and this edge contributes $+1$ to the length-rotation of $u_i$ and $-1$ to the length-rotation of $u_i'$. This means that there are exactly four vertices in $\mathcal{H}_{d+1}$ with length-rotation $(+1,\overline{v})$ and exactly four vertices with length-rotation $(-1,\overline{v})$. Since $\overline{v}\in \{+1,-1\}^{d-2}$ was arbitrary, this finishes the proof.
\end{proof}


Lastly, we observe that the length-rotations of adjacent vertices in $\mathcal{H}_d$ agree up to the coordinate corresponding to the edge between them. 

\begin{restatable}{mylemma}{ngbhrshavesamerots}
\label{lem-ngbhrshavesamerots}
Let $d \ge 2$ and let $e = uv$ be an edge of length $2^{d-1} - 2^m$ in $\mathcal{H}_d$ for some $m \in \{0,1,\dots,d-2\}$. Then if $\overline{u}$ and $\overline{v}$ are length-rotations of $u$ and $v$, respectively, $\overline{u}_i = \overline{v}_i$ for all $i \in \{m+2,m+3,\dots,d-2\}$.
\end{restatable}

\begin{proof}
    Consider the step in the process of construction of $\mathcal{H}_d$ when $e$ is added. Then $u,v$ are copies of the same vertex in two copies $\mathcal{H},\mathcal{H}'$ of $\mathcal{H}_{m-1}$ and hence have the same length-rotation in $\mathcal{H}$ and $\mathcal{H}'$, respectively. The direction of any of the edges in $\mathcal{H}_m$ are not altered in any subsequent step, so the cooresponding coordinates of $\overline{u}$ and $\overline{v}$ remain the same.
\end{proof}

From now until the end of this section, we present several results about plane subgraphs of $\mathcal{H}_d$. We start by discussing the length of plane subpaths of $\mathcal{H}_d$. 

\begin{restatable}{proposition}{longpathrotatedrawing}
\label{thm-longpathrotatedrawing}
Let $d \ge 4$, then $\mathcal{H}_d$ contains a plane path of length $2d-3$.
\end{restatable}

\begin{proof}
   Assume for simplicity that $d$ is odd (otherwise the argument is similar up to some sign changes). By Lemma~\ref{cor:all-length-rotation}, there is a vertex $v_0\in \mathcal{H}_d$ with length-rotation equal to $(+1,-1,+1,-1,\dots,-1, +1)$. We begin constructing a path starting at the edge  $e= v_0v_1$ of length $2^{d-1}-1$ adjacent to $v_0$. By Lemma~\ref{lem-ngbhrshavesamerots}, we know that the length-rotations of $u_0$ and $u_1$ agree on last $d-3$ coordinates and hence the edges of length $2^{d-1}-2$ adjacent to $u_0$ and $u_1$ both point to the right of $u_0$ and $u_1$ respectively, so we can add them and obtain a plane path of length $3$. We can iterate this process a total of $d-3$ times and obtain a plane path of length $2d-5$. Let $s,t$ be the two endpoints of this path. Recall that both $s$ and $t$ have two edges of length $2^{d-1}-2^{d-2}$ adjacent to them, one of these edges has direction $+1$ and the other $-1$, adding the edges with direction $-1$ yields a path of length $2d-3$ as we desired. Figure~\ref{fig:planepath} depicts this path in $\mathcal{H}_5$.
\end{proof}

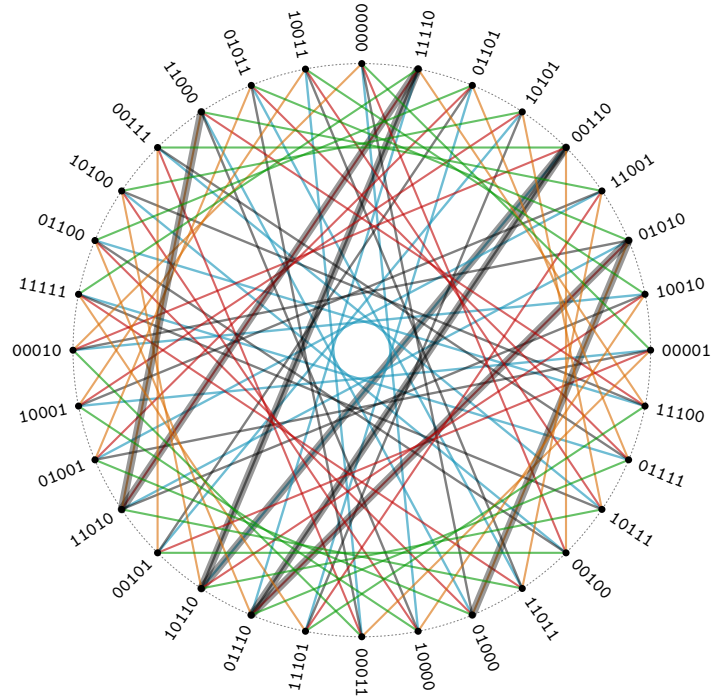
\begin{figure}
    \centering
    \begin{tikzpicture}[line cap=round,line join=round, scale=0.95]

\def\R{4}      
\def\Rl{4.5}   
\def\vdot{1.4pt}
\tikzset{
  highlight/.style={
    black,
    line width=3pt,
    draw opacity=0.4,
    line cap=round
  }
}
\draw[densely dotted,gray] (0,0) circle (\R);

\foreach \k/\lab in {0/00000,2/01101,4/00110,6/01010,8/00001,10/01111,12/00100,14/01000,16/00011,18/01110,20/00101,22/01001,24/00010,26/01100,28/00111,30/01011,
15/10000,17/11101,19/10110,21/11010,23/10001,25/11111,27/10100,29/11000,31/10011,1/11110,3/10101,5/11001,7/10010,9/11100,11/10111,13/11011}{
  \pgfmathsetmacro{\ang}{90 - 360*\k/32}
  \coordinate (v\lab) at (\ang:\R);
  \coordinate (l\lab) at (\ang:\Rl);
}

\tikzset{
e0/.style={cyan!70!black, line width=0.9pt,opacity=0.6},
  e1/.style={black, line width=0.9pt,opacity=0.5},
  e2/.style={red!75!black, line width=0.8pt,opacity=0.6},
  e3/.style={orange!85!black, line width=0.8pt,opacity=0.6},
  e4/.style={green!60!black, line width=0.8pt,opacity=0.6},
}

\draw[e0] (v00000) -- (v10000);
\draw[e0] (v00001) -- (v10001);
\draw[e0] (v00010) -- (v10010);
\draw[e0] (v00011) -- (v10011);
\draw[e0] (v00100) -- (v10100);
\draw[e0] (v00101) -- (v10101);
\draw[e0] (v00110) -- (v10110);
\draw[highlight] (v00110) -- (v10110);
\draw[e0] (v00111) -- (v10111);
\draw[e0] (v01000) -- (v11000);
\draw[e0] (v01001) -- (v11001);
\draw[e0] (v01010) -- (v11010);
\draw[e0] (v01011) -- (v11011);
\draw[e0] (v01100) -- (v11100);
\draw[e0] (v01101) -- (v11101);
\draw[e0] (v01110) -- (v11110);
\draw[e0] (v01111) -- (v11111);

\draw[e1] (v00000) -- (v01000);
\draw[e1] (v00001) -- (v01001);
\draw[e1] (v00010) -- (v01010);
\draw[e1] (v00011) -- (v01011);
\draw[e1] (v00100) -- (v01100);
\draw[e1] (v00101) -- (v01101);
\draw[e1] (v00110) -- (v01110);
\draw[highlight] (v00110) -- (v01110);
\draw[e1] (v00111) -- (v01111);

\draw[e1] (v10000) -- (v11000);
\draw[e1] (v10001) -- (v11001);
\draw[e1] (v10010) -- (v11010);
\draw[e1] (v10011) -- (v11011);
\draw[e1] (v10100) -- (v11100);
\draw[e1] (v10101) -- (v11101);
\draw[e1] (v10110) -- (v11110);
\draw[highlight] (v10110) -- (v11110);
\draw[e1] (v10111) -- (v11111);

\draw[e2] (v00000) -- (v00100);
\draw[e2] (v00001) -- (v00101);
\draw[e2] (v00010) -- (v00110);
\draw[e2] (v00011) -- (v00111);
\draw[e2] (v01000) -- (v01100);
\draw[e2] (v01001) -- (v01101);
\draw[e2] (v01010) -- (v01110);
\draw[highlight] (v01010) -- (v01110);
\draw[e2] (v01011) -- (v01111);

\draw[e2] (v10000) -- (v10100);
\draw[e2] (v10001) -- (v10101);
\draw[e2] (v10010) -- (v10110);
\draw[e2] (v10011) -- (v10111);
\draw[e2] (v11000) -- (v11100);
\draw[e2] (v11001) -- (v11101);
\draw[e2] (v11010) -- (v11110);
\draw[highlight] (v11010) -- (v11110);
\draw[e2] (v11011) -- (v11111);

\draw[e3] (v00000) -- (v00010);
\draw[e3] (v00001) -- (v00011);
\draw[e3] (v00100) -- (v00110);
\draw[e3] (v00101) -- (v00111);
\draw[e3] (v01000) -- (v01010);
\draw[highlight] (v01000) -- (v01010);
\draw[e3] (v01001) -- (v01011);
\draw[e3] (v01100) -- (v01110);
\draw[e3] (v01101) -- (v01111);

\draw[e3] (v10000) -- (v10010);
\draw[e3] (v10001) -- (v10011);
\draw[e3] (v10100) -- (v10110);
\draw[e3] (v10101) -- (v10111);
\draw[e3] (v11000) -- (v11010);
\draw[highlight] (v11000) -- (v11010);
\draw[e3] (v11001) -- (v11011);
\draw[e3] (v11100) -- (v11110);
\draw[e3] (v11101) -- (v11111);

\draw[e4] (v00000) -- (v00001);
\draw[e4] (v00010) -- (v00011);
\draw[e4] (v00100) -- (v00101);
\draw[e4] (v00110) -- (v00111);
\draw[e4] (v01000) -- (v01001);
\draw[e4] (v01010) -- (v01011);
\draw[e4] (v01100) -- (v01101);
\draw[e4] (v01110) -- (v01111);

\draw[e4] (v10000) -- (v10001);
\draw[e4] (v10010) -- (v10011);
\draw[e4] (v10100) -- (v10101);
\draw[e4] (v10110) -- (v10111);
\draw[e4] (v11000) -- (v11001);
\draw[e4] (v11010) -- (v11011);
\draw[e4] (v11100) -- (v11101);
\draw[e4] (v11110) -- (v11111);

\foreach \k/\lab in {0/00000,2/01101,4/00110,6/01010,8/00001,10/01111,12/00100,14/01000,16/00011,18/01110,20/00101,22/01001,24/00010,26/01100,28/00111,30/01011,
15/10000,17/11101,19/10110,21/11010,23/10001,25/11111,27/10100,29/11000,31/10011,1/11110,3/10101,5/11001,7/10010,9/11100,11/10111,13/11011}{

  \path (v\lab) -- (l\lab)
    node[pos=1, sloped, font=\scriptsize\ttfamily, inner sep=1pt] {\lab};

  \fill (v\lab) circle (\vdot);
}

\end{tikzpicture}
    \caption{A plane path of length $7$ in $\mathcal{H}_5$ is highlighted.}
    \label{fig:planepath}
\end{figure}
Let us now consider an alternative construction of $\mathcal{H}_d$ that will be useful for proving our next result. Let $\mathcal{H}_2'=\mathcal{H}_2$. For $d\ge 3$, we obtain $\mathcal{H}_d'$ from $\mathcal{H}_{d-1}'$ as follows. Let $v_1,v_2,\dots, v_{2^{d-1}}$ be vertices of $\mathcal{H}_{d-1}'$ ordered clockwise around $C$. Now for each vertex $v_i$ add another vertex $u_i$ between $v_i$ and $v_{i+1}$. Now for each $i\neq j$ add an edge $u_iu_j$ if and only if edge $v_iv_j$ was present in $\mathcal{H}_{d-1}'$, note that if present, the two edges $u_iu_j$ and $v_iv_j$ necessarily cross. Lastly, for each $i$, we add two parallel edges $v_iu_{i+{2^{d-2}}}$ and $v_{i+2^{d-2}}u_i$ (with indices taken modulo $2^{d-1}$, see Figure~\ref{fig:pairing}. We call the parallel edges added in the last step of construction \emph{parallel pairs}.  It is clear from the definition that if an edge $e$ belongs to a parallel pair, then it is an edge of length $2^{d-1}-1$ that crosses every edge belonging to a parallel pair different from itself. One can observe that $\mathcal{H}_d'$ is weakly isomorphic to $\mathcal{H}_d$.

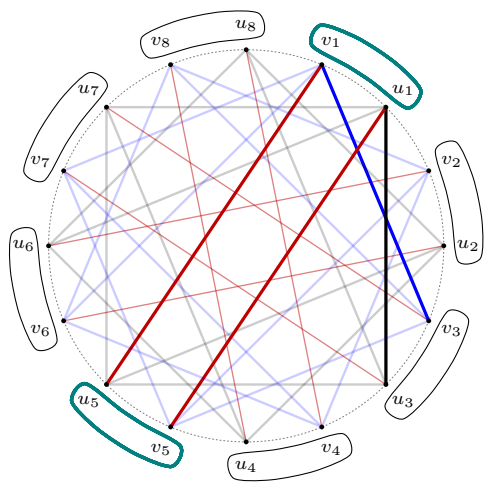
\begin{figure}[ht]
    \centering
    \begin{tikzpicture}[line cap=round,line join=round, scale=0.65]

\def\R{4}      
\def\Rl{4.5cm}   
\def\Rp{4.8cm} 
\def\Rs{4.71cm} 
\def\Rd{4.2cm}
\def\Ri{4.75cm}
\def\Rf{4.3cm}

\def\vdot{1.4pt}

\draw[densely dotted,gray] (0,0) circle (\R);

\foreach \k/\lab in {0/0000,1/1101,2/0110,3/1010,4/0001,5/1111,6/0100,7/1000,8/0011,9/1110,10/0101,11/1001,12/0010,13/1100,14/0111,15/1011}{
  \pgfmathsetmacro{\ang}{90 - 360*\k/16}
 
  \coordinate (v\lab) at (\ang:\R);
  \coordinate (l\lab) at (\ang:\Rl);
  
}

\tikzset{
  e1/.style={black, line width=0.9pt,opacity=0.2},
  e2/.style={red!75!black, line width=0.5pt, opacity=0.5},
   e5/.style={blue, line width=0.9pt, opacity=0.2},
}

\draw[e2] (v0000) -- (v1000);
\draw[e2] (v0001) -- (v1001);
\draw[e2] (v0010) -- (v1010);
\draw[e2] (v0011) -- (v1011);
\draw[e2] (v0100) -- (v1100);
\draw[e2] (v0101) -- (v1101);
\draw[e2] (v0110) -- (v1110);
\draw[e2] (v0111) -- (v1111);

\draw[e1] (v0000) -- (v0100);
\draw[e1] (v0001) -- (v0101);
\draw[e1] (v0010) -- (v0110);
\draw[e1] (v0011) -- (v0111);
\draw[e5] (v1000) -- (v1100);
\draw[e5] (v1001) -- (v1101);
\draw[e5] (v1010) -- (v1110);
\draw[e5] (v1011) -- (v1111);

\draw[e1] (v0000) -- (v0010);
\draw[e1] (v0001) -- (v0011);
\draw[e1] (v0100) -- (v0110);
\draw[e1] (v0101) -- (v0111);
\draw[e5] (v1000) -- (v1010);
\draw[e5] (v1001) -- (v1011);
\draw[e5] (v1100) -- (v1110);
\draw[e5] (v1101) -- (v1111);

\draw[blue, very thick] (v1101) -- (v1111);

\draw[e1] (v0000) -- (v0001);
\draw[black, very thick] (v0110) -- (v0100);
\draw[red!75!black, very thick] (v0110) -- (v1110);
\draw[red!75!black, very thick] (v1101) -- (v0101);

\draw[e1] (v0010) -- (v0011);
\draw[e1] (v0100) -- (v0101);
\draw[e1] (v0110) -- (v0111);
\draw[e5] (v1000) -- (v1001);
\draw[e5] (v1010) -- (v1011);
\draw[e5] (v1100) -- (v1101);
\draw[e5] (v1110) -- (v1111);

\foreach \k/\lab in {8/0000,1/0110,2/0001,3/0100,4/0011,5/0101,6/0010,7/0111}{

  \path (v\lab) -- (l\lab)
    node[pos=1,  font=\scriptsize\ttfamily, inner sep=1pt] {$u_{\k}$};

  \fill (v\lab) circle (\vdot);
}
\foreach \k/\lab in {1/1101,2/1010,3/1111,4/1000,5/1110,6/1001,7/1100,8/1011}{

  \path (v\lab) -- (l\lab)
    node[pos=1, font=\scriptsize\ttfamily, inner sep=1pt] {$v_{\k}$};

  \fill (v\lab) circle (\vdot);
}

\foreach \i in {0,...,7}{

   \pgfmathsetmacro{\angi}{67- 360*(\i)/8- 360*4/55}
   \pgfmathsetmacro{\angul}{67 -360*(\i)/8 +360/72}
 \pgfmathsetmacro{\angio}{67- 360*(\i)/8- 360*4/57.5}
   \pgfmathsetmacro{\angulo}{67 -360*(\i)/8 +360/82}
   
    \pgfmathsetmacro{\anguli}{67 -360*(\i)/8 -360*2/45}
     \pgfmathsetmacro{\angulir}{67 -360*(\i)/8 -360/70}
  \pgfmathsetmacro{\angulio}{67 -360*(\i)/8 -360*2/50}
\pgfmathsetmacro{\anguliro}{67 -360*(\i)/8 -360/55}

  \coordinate (a\i) at ({\angi}:\Rf);
   \coordinate (d\i) at ({\angul}:\Rf);
  \coordinate (h\i) at ({\angio}:\Ri);
   \coordinate (e\i) at ({\angulo}:\Ri);
   
    \coordinate (f\i) at ({\angulir}:\Rp);
   \coordinate (g\i) at ({\anguli}:\Rp);
   
     \coordinate (b\i) at ({\angulio}:\Rd);
     \coordinate (c\i) at ({\anguliro}:\Rd);
   }

\foreach \i in {0,...,7}{
   \draw plot [smooth cycle] coordinates {
  (d\i) (e\i) (f\i) (g\i) (h\i) (a\i) (b\i) (c\i) 
};
 \draw[teal, very thick] plot [smooth cycle] coordinates {
  (d0) (e0) (f0) (g0) (h0) (a0) (b0) (c0) };
   \draw[teal, very thick] plot [smooth cycle, red] coordinates {
  (d4) (e4) (f4) (g4) (h4) (a4) (b4) (c4) };
 }

\end{tikzpicture}
    \caption{The alternative construction of $\mathcal{H}_d$.}
    \label{fig:pairing}
\end{figure}


To show that Proposition~\ref{thm-longpathrotatedrawing} is in fact optimal, we prove the following more general result.

\planesubgraphrotateddrawing*

\begin{proof}
    We proceed by induction on $d$. If $d=3$, the statement can be checked by an easy case distinction. We assume that $d\ge 4$ and that the statement holds for $d-1$. Now we distinguish two cases based on whether $G$ contains some edge of maximal length in $\mathcal{H}_d$ ($2^{d-1}-1$) or not. Assume that the latter is true and let $v_1,v_2,\dots, v_{2^{d-1}}$ and $u_1,u_2, \dots, u_{2^{d-1}}$ be as in the second definition of $H_d$. Recall that for any $i\neq j$ there are either zero or exactly two edges of length smaller than $2^{d-1}-1$ between $(u_i,v_i)$ and $(u_j,v_j)$ and if those edges exist then they necessarily cross. Therefore, $G$ can contain at most one edge connecting a vertex in $(u_i,v_i)$ to a vertex in $(u_j,v_j)$ for each $i\neq j$. Lastly, note that if we replace each pair $(u_i,v_i)$ by a single vertex $w_i$ and we connect $w_i$ to $w_j$ if and only if there was an edge shorter than $\frac{1}{2}-\frac{1}{2^d}$ between $(u_i,v_i)$ and $(u_j,v_j)$ we obtain a drawing of $Q_{d-1}$ which is the same as $\mathcal{H}_{d-1}$ (up to some rotation of the plane). Therefore, the subgraph of $G$ is now a plane subgraph of $\mathcal{H}_{d-1}$ and therefore contains at most $2(d-1)-2 = 2d-4$ edges by our inductive assumption. Otherwise, note that $G$ can contain at most $2$ edges of length $2^{d-1}-1$ as every such edge belongs to a parallel pair and is crossed by all but one other edge of the same length. Removing these edges places us in the first case, and therefore $G$ has at most $2(d-1) -2 +2 = 2d-2$ edges. This finishes the proof.
\end{proof}

We obtain Theorem~\ref{thm:drawingswithnolongpaths} as a consequence of Theorem~\ref{thm-planesubgraphrotateddrawing}. 
\drawingswithnolongpaths*
\begin{proof}
   A plane path in $\mathcal{H}_d$ contains at most one edge of length $2^{d-1}-1$. Then the result follows by the same argument as in Theorem~\ref{thm-planesubgraphrotateddrawing}.  
\end{proof}

We now prove Theorem~\ref{thm:drawingswithnolargematchings}.

\drawingswithnolargematchings*

\begin{proof}
For an edge $e$ of length $\ell$ in any convex drawing of a graph with $n$ vertices, a matching cannot contain more than $\lfloor\frac{n}{l+1}\rfloor$ edges of length $l$. Let $M$ be a plane matching in $\mathcal{H}_d$. The length profile in $\mathcal{H}_d$ is $\{2^{d-1}-1, 2^{d-1}-2,2^{d-1}-4, \cdots, 2^{d-1}-2^{d-3}, 2^{d-1}-2^{d-2},2^{d-1}-2^{d-2}\}$. Therefore, $M$ can contain at most $3$ edges of length $2^{d-1}-2^{d-2}$ and $2$ edges of any other length.

   \textbf{Case 1: $M$ contains $3$ edges of length $2^{d-1}-2^{d-2}$.} We can show that in this case $|E(M)|\le d-1$. First, we prove that the existence of 3 edges of length $2^{d-1}-2^{d-2}$ forbids the matching to have edges of length $2^{d-1}-1$ and $2^{d-1}-2$.  Any edge $j=vw$ in the drawing has two edges of length $2^{d-1}-2^{d-2}$ on one side of $j$, and one on the other side of $j$. Then one of $[v,w]$ or $[w,v]$ contains two of the three edges of length $2^{d-1}-2^{d-2}$ and thus at least  $2(2^{d-1}-2^{d-2})+3=2^{d-1}+3$ vertices of $\mathcal{H}_d$. This implies that the length of $j$ is at most $2^{d-1}-3$ which proves the claim; see the left picture in Figure~\ref{fig:plane matching cases}. Next, we show that in this case, no edge of length $l$ for $l\in \{2^{d-1}-4, 2^{d-1}-8, \dots, 2^{d-1}-2^{d-3}\}$ can appear more than once. The next length in our length profile that is higher than $2^{d-1}-2^{d-2}$ is $2^{d-1}-2^{d-3}$. By a similar argument to the one above, we can see that we cannot have two edges of length at least $2^{d-1}-2^{d-3}$ together with three copies of $2^{d-1}-2^{d-2}$ in the plane matching, see the middle picture in Figure~\ref{fig:plane matching cases}. This proves the claim in Case~1: the size of the plane matching cannot exceed $d-1$ ($3$ copies of the shortest length plus one copy of each length in $\{2^{d-1}-2^{d-3},\cdots,2^{d-1}-2^{2}\}$). 

\begin{figure}[t]
\centering
\begin{tikzpicture}[line cap=round,line join=round, scale=0.45]
\begin{scope}[xshift=-10cm]
    \def\R{4}      
\def\Rl{4.5cm}   
\def\vdot{1.4pt}
\draw[gray] (0,0) circle (\R);
\foreach \k in {0,...,31}{
  \pgfmathsetmacro{\ang}{90 - 360*\k/32}
  \coordinate (v\k) at (\ang:\R);
  \coordinate (l\k) at (\ang:\Rl);
    \fill (v\k) circle (\vdot);
}
\tikzset{
  e1/.style={black, line width=1.1pt},
    e2/.style={blue, dotted, line width=1.1pt}
}
\draw[e1] (v0) -- (v8);
\draw[e1] (v9) -- (v17);
\draw[e1] (v21) -- (v29);
\draw[e2] (v31) -- (v18);
\node  at (-0.6,-.3) {\small$j$};
\end{scope}
\def\R{4}      
\def\Rl{4.5cm}   
\def\vdot{1.4pt}
\draw[gray] (0,0) circle (\R);
\foreach \k in {0,...,31}{
  \pgfmathsetmacro{\ang}{90 - 360*\k/32}
  \coordinate (v\k) at (\ang:\R);
  \coordinate (l\k) at (\ang:\Rl);
    \fill (v\k) circle (\vdot);
}
\tikzset{
  e1/.style={black, line width=1.1pt},
    e2/.style={blue, line width=1.1pt},
        e3/.style={black, dotted, line width=1.1pt}
}
\draw[e1] (v0) -- (v8);
\draw[e1] (v11) -- (v19);
\draw[e2] (v9) -- (v29);
\draw[e2] (v10) -- (v22);
\draw (l29) arc[start angle=90 + 360*3/32, end angle=90 + 360*10/32, radius=\Rl];
\node at (-4,3) {$i$};

\begin{scope}[xshift=10cm]
    
\def\R{4}      
\def\Rl{4.5cm}   
\def\vdot{1.4pt}
\draw[gray] (0,0) circle (\R);
\foreach \k in {0,...,31}{
  \pgfmathsetmacro{\ang}{90 - 360*\k/32}
  \coordinate (v\k) at (\ang:\R);
  \coordinate (l\k) at (\ang:\Rl);
    \fill (v\k) circle (\vdot);
  
}
\tikzset{
  e1/.style={red, line width=1.1pt},
    e2/.style={blue, line width=1.1pt},
        e3/.style={black, dotted, line width=1.1pt}
}
\draw[e1] (v0) -- (v15);
\draw[e2] (v0) -- (v14);
\draw[e2] (v15) -- (v29);
\draw[e2] (v16) -- (v30);
\draw[e1] (v16) -- (v31);
\node at (1,1) {$e$};
\end{scope}
\end{tikzpicture}

    \caption{The cases in the proof of Theorem~\ref{thm:drawingswithnolargematchings}. The black edges have length $2^{d-1}-2^{d-2}$. Left: Showing that in Case~1, an edge $j$ in the plane matching must have length less than $2^{d-1}-2$. Middle: Showing that in Case~1, no other edge length can appear more than once in the plane matching. The two blue lines have length at least $2^{d-1}-2^{d-3}$ and therefore the arc $i$ has length less than $2^{d}-(2(2^{d-1}-2^{d-3})+1)=2^{d-2}-1$. Right: Proof of Case~2, showing that among the depicted lines (2 copies of the two longest lengths) only two can belong to the matching.}
    \label{fig:plane matching cases}
\end{figure}
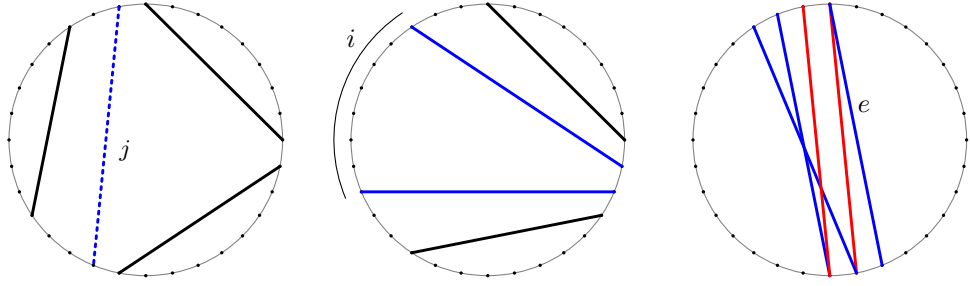
   \textbf{Case 2: $M$ contains at most $2$ edges of length $2^{d-1}-2^{d-2}$.} We claim that $M$ contains at most two edges of length at least $2^{d-1}-2$. We assume that there is an edge $e=vw$ of length $2^{d-1}-2$ in $M$ such that $[v,w]$ contains $2^{d-1}-1$ vertices. Note that it is impossible for $M$ to contain three edges of length $2^{d-1}-1$, so if $M$ contains no edges of length $2^{d-1}-2$, we are done. If there are any other edges of length $2^{d-1}-2$  or higher in $M$, both endpoints of such edges must lie in $[w,v]$. Now, there are two edges of length $2^{d-1}-1$ and two edges of length $2^{d-1}-2$ with both endpoints in $[w,v]$. However, one of the longer edges has either $w$ or $v$ as an endpoint and hence cannot be in $M$. The other edge of length $2^{d-1}-1$ either shares an endpoint with or crosses the two shorter edges, so if it is contained in $M$, none of the shorter edges can be. Lastly, two short edges cross, and hence only one of them can be in $M$. See the right picture in Figure~\ref{fig:plane matching cases} for an illustration. This finishes the proof of the claim in Case~2. 

    By Cases~1 and~2 we conclude that the size of a maximum plane matching is at most $2d-4$. Now, we show that there always exists a matching of this size in $\mathcal{H}_d$. \end{proof}

We now show that the upper bound of Theorem~\ref{thm:drawingswithnolargematchings} is achieved by $\mathcal{H}_d$. Therefore, to improve it, we would need a different construction. In order to prove this, we need the following claim, which is essentially just a stronger version of Lemma~\ref{lem-allrotationshere}. We omit the proof because it is rather tedious and technical.

\begin{remark} \label{lem-allrotsstronger}
For each vertex $x\in V(Q_d)$, its length rotation in $\mathcal{H}_d$ is given by $\overline{x}\in \{-1,+1\}^{d-2}$ with $\overline{x}_i=\begin{cases}
    +1,\text{ if }x_i=0,\\
    -1,\text{ if }x_i=1,
\end{cases}$  
for each $i\in \{1,\ldots,d-2\}$.
Thus, the vertices with length rotation $\underbrace{(+1,+1,\ldots,+1)}_{d-2\text{ times}}$ are exactly the vertices $\overline{0\ldots0x}$ with $x\in V(Q_2)$, which are the non-rotated copies of the vertices from $\mathcal{H}_2$ and thus positioned equidistantly (at distance $2^{d-2}$) around the circle.

Starting at one of those four vertices and going counterclockwise around the circle (e.g. from the vertex at position $0$ to the vertex at position $1-\frac{1}{2^{d}}$ etc.), we see that the first $d-2$ coordinates of the vertices are the binary encodings of the integers from $[0,2^{d-2}]$ in increasing order. By this, we mean that if we go $0\leq m\leq 2^{d-2}-1$ steps counterclockwise from our starting vertex $\overline{\underbrace{0\ldots0}_{d-2}x}$, we obtain a vertex $y_m$ with $m=\sum_{i=1}^{d-2}y_i2^{i-1}$.



\end{remark}


\begin{figure}
    \centering
    \begin{tikzpicture}[line cap=round,line join=round, scale=0.85]

\def\R{4}      
\def\Rl{4.5}   
\def\Rn{5.2}
\def\vdot{1.4pt}
\tikzset{
  highlight/.style={
    black,
    line width=3pt,
    draw opacity=0.4,
    line cap=round
  }
}
\draw[densely dotted,gray] (0,0) circle (\R);

\foreach \k/\lab in {0/00000,2/01101,4/00110,6/01010,8/00001,10/01111,12/00100,14/01000,16/00011,18/01110,20/00101,22/01001,24/00010,26/01100,28/00111,30/01011,
15/10000,17/11101,19/10110,21/11010,23/10001,25/11111,27/10100,29/11000,31/10011,1/11110,3/10101,5/11001,7/10010,9/11100,11/10111,13/11011}{
  \pgfmathsetmacro{\ang}{90 - 360*\k/32}
  \coordinate (v\lab) at (\ang:\R);
  \coordinate (l\lab) at (\ang:\Rl);
   \coordinate (n\lab) at (\ang:\Rn);
}

\tikzset{
e0/.style={cyan!70!black, line width=0.9pt,opacity=0.6},
  e1/.style={black, line width=0.9pt,opacity=0.5},
  e2/.style={red!75!black, line width=0.8pt,opacity=0.6},
  e3/.style={orange!85!black, line width=0.8pt,opacity=0.6},
  e4/.style={green!60!black, line width=0.8pt,opacity=0.6},
}

\draw[e0] (v00000) -- (v10000);
\draw[e0] (v00001) -- (v10001);
\draw[e0] (v00010) -- (v10010);
\draw[e0] (v00011) -- (v10011);
\draw[e0] (v00100) -- (v10100);
\draw[e0] (v00101) -- (v10101);
\draw[e0] (v00110) -- (v10110);
\draw[e0] (v00111) -- (v10111);
\draw[e0] (v01000) -- (v11000);
\draw[highlight,cyan] (v01000) -- (v11000);
\draw[e0] (v01001) -- (v11001);
\draw[e0] (v01010) -- (v11010);
\draw[e0] (v01011) -- (v11011);
\draw[highlight, cyan] (v01011) -- (v11011);
\draw[e0] (v01100) -- (v11100);
\draw[e0] (v01101) -- (v11101);
\draw[e0] (v01110) -- (v11110);
\draw[e0] (v01111) -- (v11111);

\draw[e1] (v00000) -- (v01000);
\draw[e1] (v00001) -- (v01001);
\draw[e1] (v00010) -- (v01010);
\draw[e1] (v00011) -- (v01011);
\draw[e1] (v00100) -- (v01100);
\draw[e1] (v00101) -- (v01101);
\draw[e1] (v00110) -- (v01110);
\draw[e1] (v00111) -- (v01111);

\draw[e1] (v10000) -- (v11000);
\draw[highlight] (v10000) -- (v11000);
\draw[e1] (v10001) -- (v11001);
\draw[e1] (v10010) -- (v11010);
\draw[e1] (v10011) -- (v11011);
\draw[highlight] (v10011) -- (v11011);
\draw[e1] (v10100) -- (v11100);
\draw[e1] (v10101) -- (v11101);
\draw[e1] (v10110) -- (v11110);
\draw[e1] (v10111) -- (v11111);

\draw[e2] (v00000) -- (v00100);
\draw[highlight] (v00000) -- (v00100);
\draw[e2] (v00001) -- (v00101);
\draw[e2] (v00010) -- (v00110);
\draw[e2] (v00011) -- (v00111);
\draw[highlight] (v00011) -- (v00111);
\draw[e2] (v01000) -- (v01100);
\draw[e2] (v01001) -- (v01101);
\draw[e2] (v01010) -- (v01110);
\draw[e2] (v01011) -- (v01111);

\draw[e2] (v10000) -- (v10100);
\draw[e2] (v10001) -- (v10101);
\draw[e2] (v10010) -- (v10110);
\draw[e2] (v10011) -- (v10111);
\draw[e2] (v11000) -- (v11100);
\draw[e2] (v11001) -- (v11101);
\draw[e2] (v11010) -- (v11110);
\draw[e2] (v11011) -- (v11111);

\draw[e3] (v00000) -- (v00010);
\draw[e3] (v00001) -- (v00011);
\draw[e3] (v00100) -- (v00110);
\draw[e3] (v00101) -- (v00111);
\draw[e3] (v01000) -- (v01010);
\draw[e3] (v01001) -- (v01011);
\draw[e3] (v01100) -- (v01110);
\draw[e3] (v01101) -- (v01111);

\draw[e3] (v10000) -- (v10010);
\draw[e3] (v10001) -- (v10011);
\draw[e3] (v10100) -- (v10110);
\draw[e3] (v10101) -- (v10111);
\draw[e3] (v11000) -- (v11010);
\draw[e3] (v11001) -- (v11011);
\draw[e3] (v11100) -- (v11110);
\draw[highlight] (v11100) -- (v11110);
\draw[e3] (v11101) -- (v11111);
\draw[highlight] (v11101) -- (v11111);

\draw[e4] (v00000) -- (v00001);
\draw[e4] (v00010) -- (v00011);
\draw[e4] (v00100) -- (v00101);
\draw[e4] (v00110) -- (v00111);
\draw[e4] (v01000) -- (v01001);
\draw[e4] (v01010) -- (v01011);
\draw[e4] (v01100) -- (v01101);
\draw[e4] (v01110) -- (v01111);

\draw[e4] (v10000) -- (v10001);
\draw[e4] (v10010) -- (v10011);
\draw[e4] (v10100) -- (v10101);
\draw[e4] (v10110) -- (v10111);
\draw[e4] (v11000) -- (v11001);
\draw[e4] (v11010) -- (v11011);
\draw[e4] (v11100) -- (v11101);
\draw[e4] (v11110) -- (v11111);

\foreach \k/\lab in {0/00000,2/01101,4/00110,6/01010,8/00001,10/01111,12/00100,14/01000,16/00011,18/01110,20/00101,22/01001,24/00010,26/01100,28/00111,30/01011,
15/10000,17/11101,19/10110,21/11010,23/10001,25/11111,27/10100,29/11000,31/10011,1/11110,3/10101,5/11001,7/10010,9/11100,11/10111,13/11011}{

  \path (v\lab) -- (l\lab)
    node[pos=1, sloped, font=\scriptsize\ttfamily, inner sep=1pt] {\lab};

  \fill (v\lab) circle (\vdot);
}
 \path (v01000) -- (n01000)
    node[pos=1, font=\footnotesize\ttfamily, inner sep=1pt] {$v_0$};
 \path (v10000) -- (n10000)
    node[pos=1, font=\footnotesize\ttfamily, inner sep=1pt] {$v_1$};
 \path (v00011) -- (n00011)
    node[pos=1, font=\footnotesize\ttfamily, inner sep=1pt] {$v_2$};   
\path (v11101) -- (n11101)
    node[pos=1, font=\footnotesize\ttfamily, inner sep=1pt] {$v_3$}; 
 \path (v01011) -- (n01011)
    node[pos=1, font=\footnotesize\ttfamily, inner sep=1pt] {$w_0$};
 \path (v10011) -- (n10011)
    node[pos=1, font=\footnotesize\ttfamily, inner sep=1pt] {$w_1$};
 \path (v00000) -- (n00000)
    node[pos=1, font=\footnotesize\ttfamily, inner sep=1pt] {$w_2$};   
\path (v11110) -- (n11110)
    node[pos=1, font=\footnotesize\ttfamily, inner sep=1pt] {$w_3$};     

\end{tikzpicture}
    \caption{A plane matching of length $6$ in $\mathcal{H}_5$ from the proof of Proposition~\ref{prop:existencematching} is highlighted in grey and vertices $v_0,\dots,v_3$ and $w_0,\dots,w_3$ are marked. The two edges highlighted in cyan form a parallel pair that is uncrossed by the matching and can be added to form a larger plane subgraph.}
    \label{fig:planematching}
\end{figure}
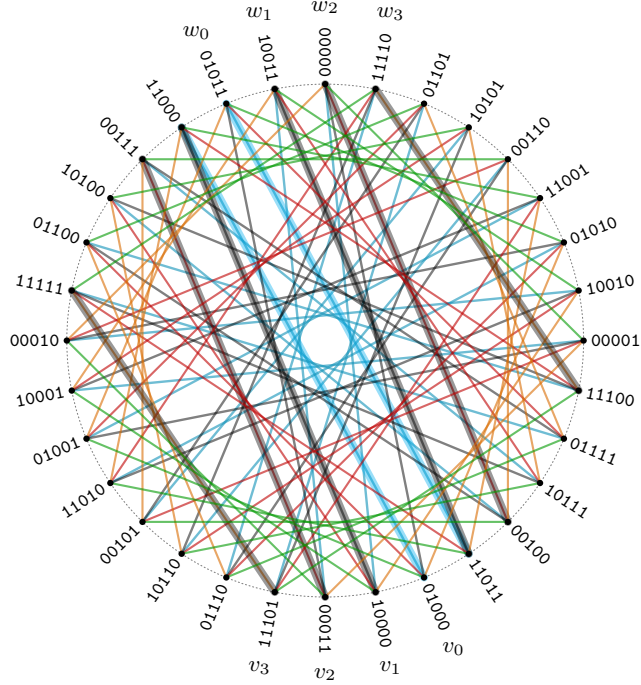

\begin{proposition}\label{prop:existencematching}
For every $d\ge 3$, $\mathcal{H}_d$ contains a plane matching on $2d-4$ edges and a plane subgraph on $2d-2$ edges. 
\end{proposition}

\begin{proof}   
Our construction is based on a claim about non-negative integers.
\begin{claim} \label{claim:numbertheory}
For each $k\in\mathbb{N}_{\geq 1}$ there is an element $x_k\in\{0,\ldots,2^{k}-1\}$ such that for each $i\in \{0,\ldots,k-1\}$:
$x_k-i\in \{0,1,\ldots,2^{i}-1\}+2^{i+1}\mathbb{N}_0$, i.e., the coefficient of $2^{i}$ in the binary representation of $x_k-i$ is $0$.\end{claim}
\begin{claimproof}
For $k=1$ and thus $i=0$, we have  $x_1=0\in \{0\}+2\mathbb{N}_0$.
Now let us assume the claim is true for some $k\in\mathbb{N}_{\geq 1}$ and let $x_k$ be such an element.
As $x_k\in\{0,\ldots,2^{k}-1\}$ and $x_k-(k-1)\in \{0,1,\ldots,2^{k-1}-1\}+2^{k}\mathbb{N}_0$, we have $x_k-k\in \{-1,0,1,\ldots,2^{k-1}-2\}$.
If $x_k-k\in \{0,1,\ldots,2^{k-1}-2\}\subseteq \{0,1,\ldots,2^{k}-1\}+2^{k+1}\mathbb{N}_0$, we see that $x_{k+1}:=x_k$ fulfills the claim (for the missing $i=k$ we had to check).
Otherwise we have $x_k=k-1$ and we set $x_{k+1}:=k-1+2^{k}$.
For $i\in \{0,\ldots,k-1\}$, we have $x_{k+1}-i=x_k-i+2^{k}\in \{0,1,\ldots,2^{i}-1\}+2^{i+1}\mathbb{N}_0$
and for $i=k$, we have
$x_{k+1}-i=2^{k}-1\in \{0,1,\ldots,2^{k}-1\}+2^{k+1}\mathbb{N}_0.$ Thus, we have inductively proven the claim. \end{claimproof}

We now construct our matching. Let $x_{d-2}$ as in the Claim~\ref{claim:numbertheory}, we will identify natural numbers with their binary representation in what follows. Let $v_0$ be the vertex of $\mathcal{H}_d$ corresponding to $\overline{x_{d-2}00}$ and $w_0$ be the vertex of $\mathcal{H}_d$ corresponding to $\overline{x_{d-2}11}$, which are antipodal by Remark~\ref{lem-allrotsstronger}. Note that the two edges of length $2^{d-1}-1$ adjacent to $v_0$ and $w_0$ induce a parallel pair. Now, let $v_1,v_2,\dots, v_{d-2}$ be vertices of $\mathcal{H}_d$ consecutive on the circle $C$ when traversing it clockwise from $v_0$ and define $w_1,w_2,\dots w_{d-2}$, analogously. Now let $M_v$ be the set of edges containing the second longest (i.e., of length $2^{d-1}-2$) edge adjacent to $v_1$, third longest (i.e., of length $2^{d-1}-2^2$) edge adjacent to $v_2$ and so on. That is, for vertex $v_i$, we include the edge of length $2^{d-1}-2^i$ incident to $v_i$ in $M_v$. See Figure~\ref{fig:planematching} for the construction of $M_v$.  By Remark~\ref{lem-allrotsstronger} and Claim~\ref{claim:numbertheory}, we know that each edge in $M_v$ contributes $+1$ to the length rotation of the vertex $v_i$ to which it is adjacent \footnote{This is not technically correct for $v_{d-2}$ as it is adjacent to an edge of length $2^{d-1}-2^{d-2}$ which is not included in the length-rotation but as there are two such edges, we can assume that we are using the one that would contribute $+1$ (i.e., the one that is to the left of~$v_{d-2}$)}. We show that $M_v$ is a plane matching on $d-2$ edges. Assume that there are two edges $v_iv_i'$ and $v_jv_j'$ in $M_v$ for $1\le i < j \le d-2$ that are not disjoint (i.e., they either cross or share a vertex). If these edges are not disjoint, it means that $v_i'$ lies inside the interval $(v_j,v_j']$. Note that this interval contains $2^{d-1}- 2^j$ vertices of $\mathcal{H}_d$. Furthermore, $(v_i,v_j]$ contains at most $j-i$ vertices of $\mathcal{H}_d$. Therefore, the interval $(v_i,v_j']$ can contain at most $2^{d-1} - 2^j + j - i$ vertices. On the other hand, by the definition of $M_v$ the edge $v_iv_i'$ is of length $2^{d-1}-2^i$ and the interval $(v_i,v_i']$ contains exactly $2^{d-1} - 2^i$ vertices of $\mathcal{H}_d$. Hence, $v_i'$ can lie inside the interval $(v_j,v_j']$ only if $2^{d-1} - 2^j + j - i \ge  2^{d-1} - 2^i.$ This holds if and only if $2^j - 2^i \le j-i $ which is impossible if both $i,j$ are positive. Therefore, $M_v$ is a plane matching on $d-2$ edges.

 Analogously to $M_v$, we define $M_w$ using the vertices $w_1,w_2,\dots w_{d-2}$. It is again a plane matching on $d-2$ edges by the same arguments as before. Furthermore, by the definition of $M_v$ and $M_w$, none of the edges crosses the parallel pair induced by the edges of length $2^{d-1}-1$ adjacent to $v_0$ and $w_0$, and therefore, the disjoint union $M_v \cup M_w$ is a plane matching on $2d-4$ edges as we wanted. Additionally, $M_v \cup M_w\cup\{v_0v_0',w_0w_0'\}$ induces a plane subgraph with $2d-2$ edges. \qedhere





\end{proof}

Lastly, we give a necessary condition for an abstract graph to be embeddable  into every rectilinear drawing of the $d$-cube (for a sufficiently large $d$).

\graphsembeddableinQd*

\begin{proof}
 We first show that $G$ must be acyclic. Since $Q_d$ is bipartite, consider a bipartition of it and let $A_d,B_d$ be the two parts. Now, consider a drawing  $\mathcal{G}$ of $Q_d$ in which the vertices of $A_d$ and $B_d$ are placed on a circle in such a way that there is a line $\ell$ such that all vertices of $A_d$ lie on one side of $\ell$ and all vertices of $B_d$ lie on the other side. Now, if there is a plane cycle $\mathcal{C}$ of length $k\ge 4$ in $\mathcal{G}$, we can label its vertices $a_1,b_1,a_2,b_2,\dots, a_k,b_k$ in such a way that for each $i$, $a_ib_i$, and $b_ia_{i+1}$ are edges of $\mathcal{C}$. Consider the edges $a_1b_1$, $b_1a_2$, and $b_ka_1$. Since $\mathcal{C}$ is  plane, all of the vertices $b_2,b_3,\dots,b_k$ must lie in the same halfplane determined by the line through $a_2$ and $b_1$. However, $a_1$ must lie in the other halfplane determined by this line; otherwise, edges $a_1b_1$ and $a_2b_2$ would cross. However, this means that the edge $b_ka_1$ must cross the edge $b_1a_2$. Therefore, $\mathcal{G}$ does not contain any plane cycles. 
 Now, to prove that $G$ must be a forest of caterpillars, it suffices to prove that for all $d\ge 2$, $\mathcal{H}_d$ does not contain the graph $G_0$, obtained by  subdividing each edge of $K_{1,3}$ once, as a plane subgraph.  Assume for contradiction that there is some $d$ such that $\mathcal{H}_d$ contains a plane copy of $G_0$. Let $v$ be the unique degree three vertex in $G_0$, and let $e_1=vw_1, e_2=vw_2, e_3=vw_3$ be the three edges adjacent to $v$ in $G_0$ and let $w_1',w_2'$ and $w_3'$ be the leaves in $G_0$ adjacent to $w_1,w_2$ and $w_3$, respectively. Without loss of generality, we can assume that in $\mathcal{H}_d$, $G_0$ is embedded so that $e_1,e_2,e_3$ appear in this order clockwise around $v$. 
 Now, up to symmetry, the remaining edges of $G_0$ are embedded in one of the three ways depicted in Figure~\ref{fig:3configurations}. For the remainder of the proof, we say that a vertex $w_i'$ lies to the left of $w_i$ if it lies in the interval $(w_i,v)$, and that it lies to the right of $w_i$ if it lies in the interval $(v,w_i)$.

 \begin{figure}[ht]
     \centering
     \begin{minipage}{0.3\textwidth}
         \begin{tikzpicture}[line cap=round,line join=round, scale=0.45]

\def\R{4}      
\def\Rl{4.5cm}   
\def\vdot{1.4pt}
\draw[gray] (0,0) circle (\R);
\foreach \k in {0,...,30}{
  \pgfmathsetmacro{\ang}{90 - 360*\k/31}
  \coordinate (v\k) at (\ang:\R);
  \coordinate (l\k) at (\ang:\Rl);
}
\tikzset{
  e1/.style={black, line width=0.9pt}
}
\draw[e1] (v0) -- (v8);
\draw[e1] (v0) -- (v15);
\draw[e1] (v0) -- (v22);
\draw[e1] (v8) -- (v13);
\draw[e1] (v15) -- (v20);
\draw[e1] (v22) -- (v28);
\path (v0) -- (l0)
    node[pos=1,  font=\scriptsize\ttfamily, inner sep=1pt] {$v$};
    \path (v8) -- (l8)
    node[pos=1,  font=\scriptsize\ttfamily, inner sep=1pt] {$w_1$};
    \path (v13) -- (l13)
    node[pos=1,  font=\scriptsize\ttfamily, inner sep=1pt] {$w_1'$};
    \path (v15) -- (l15)
    node[pos=1,  font=\scriptsize\ttfamily, inner sep=1pt] {$w_2$};
    \path (v20) -- (l20)
    node[pos=1,  font=\scriptsize\ttfamily, inner sep=1pt] {$w_2'$};
    \path (v22) -- (l22)
    node[pos=1,  font=\scriptsize\ttfamily, inner sep=1pt] {$w_3$};
        \path (v28) -- (l28)
    node[pos=1,  font=\scriptsize\ttfamily, inner sep=1pt] {$w_3'$};
    \fill (v0) circle (\vdot);
    \fill (v8) circle (\vdot);
    \fill (v13) circle (\vdot);
    \fill (v15) circle (\vdot);
    \fill (v20) circle (\vdot);
    \fill (v22) circle (\vdot);
    \fill (v28) circle (\vdot);
    \node  at (1.5,1.8) {\small$e_1$};
    \node  at (-0.4,-.3) {\small$e_2$};
    \node  at (-1.7,1) {\small$e_3$};
\end{tikzpicture}
     \end{minipage}
     \hfill
           \begin{minipage}{0.3\textwidth}

           \begin{tikzpicture}[line cap=round,line join=round, scale=0.45]

\def\R{4}      
\def\Rl{4.5cm}   
\def\vdot{1.4pt}
\draw[gray] (0,0) circle (\R);
\foreach \k in {0,...,30}{
  \pgfmathsetmacro{\ang}{90 - 360*\k/31}
  \coordinate (v\k) at (\ang:\R);
  \coordinate (l\k) at (\ang:\Rl);
}
\tikzset{
  e1/.style={black, line width=0.9pt}
}
\draw[e1] (v0) -- (v8);
\draw[e1] (v0) -- (v20);
\draw[e1] (v0) -- (v22);
\draw[e1] (v13) -- (v8);
\draw[e1] (v20) -- (v15);
\draw[e1] (v22) -- (v28);
\path (v0) -- (l0)
    node[pos=1,  font=\scriptsize\ttfamily, inner sep=1pt] {$v$};
    \path (v8) -- (l8)
    node[pos=1,  font=\scriptsize\ttfamily, inner sep=1pt] {$w_1$};
    \path (v13) -- (l13)
    node[pos=1,  font=\scriptsize\ttfamily, inner sep=1pt] {$w_1'$};
    \path (v20) -- (l20)
    node[pos=1,  font=\scriptsize\ttfamily, inner sep=1pt] {$w_2$};
    \path (v15) -- (l15)
    node[pos=1,  font=\scriptsize\ttfamily, inner sep=1pt] {$w_2'$};
    \path (v22) -- (l22)
    node[pos=1,  font=\scriptsize\ttfamily, inner sep=1pt] {$w_3$};
        \path (v28) -- (l28)
    node[pos=1,  font=\scriptsize\ttfamily, inner sep=1pt] {$w_3'$};
    \fill (v0) circle (\vdot);
    \fill (v8) circle (\vdot);
    \fill (v13) circle (\vdot);
    \fill (v15) circle (\vdot);
    \fill (v20) circle (\vdot);
    \fill (v22) circle (\vdot);
    \fill (v28) circle (\vdot);
        \node  at (1.5,1.8) {\small$e_1$};
    \node  at (-1.4,-.3) {\small$e_2$};
    \node  at (-1.9,2.3) {\small$e_3$};
\end{tikzpicture}
      \end{minipage}
      \hfill
      \begin{minipage}{0.3\textwidth}
         \begin{tikzpicture}[line cap=round,line join=round, scale=0.45]

\def\R{4}      
\def\Rl{4.5cm}   
\def\vdot{1.4pt}
\draw[gray] (0,0) circle (\R);
\foreach \k in {0,...,30}{
  \pgfmathsetmacro{\ang}{90 - 360*\k/31}
  \coordinate (v\k) at (\ang:\R);
  \coordinate (l\k) at (\ang:\Rl);
}
\tikzset{
  e1/.style={black, line width=0.9pt}
}
\draw[e1] (v0) -- (v12);
\draw[e1] (v0) -- (v15);
\draw[e1] (v0) -- (v22);
\draw[e1] (v12) -- (v4);
\draw[e1] (v15) -- (v20);
\draw[e1] (v22) -- (v28);
\path (v0) -- (l0)
    node[pos=1,  font=\scriptsize\ttfamily, inner sep=1pt] {$v$};
    \path (v12) -- (l12)
    node[pos=1,  font=\scriptsize\ttfamily, inner sep=1pt] {$w_1$};
    \path (v4) -- (l4)
    node[pos=1,  font=\scriptsize\ttfamily, inner sep=1pt] {$w_1'$};
    \path (v15) -- (l15)
    node[pos=1,  font=\scriptsize\ttfamily, inner sep=1pt] {$w_2$};
    \path (v20) -- (l20)
    node[pos=1,  font=\scriptsize\ttfamily, inner sep=1pt] {$w_2'$};
    \path (v22) -- (l22)
    node[pos=1,  font=\scriptsize\ttfamily, inner sep=1pt] {$w_3$};
        \path (v28) -- (l28)
    node[pos=1,  font=\scriptsize\ttfamily, inner sep=1pt] {$w_3'$};
    \fill (v0) circle (\vdot);
    \fill (v4) circle (\vdot);
    \fill (v12) circle (\vdot);
    \fill (v15) circle (\vdot);
    \fill (v20) circle (\vdot);
    \fill (v22) circle (\vdot);
    \fill (v28) circle (\vdot);
    \node  at (1.5,1.8) {\small$e_1$};
    \node  at (-0.4,-.3) {\small$e_2$};
    \node  at (-1.7,1) {\small$e_3$};
\end{tikzpicture}
      \end{minipage}
      
     \caption{Three possible cases of $G_0$ in the proof of Theorem~\ref{thm:graphsembeddableinQd}}
     \label{fig:3configurations}
 \end{figure}
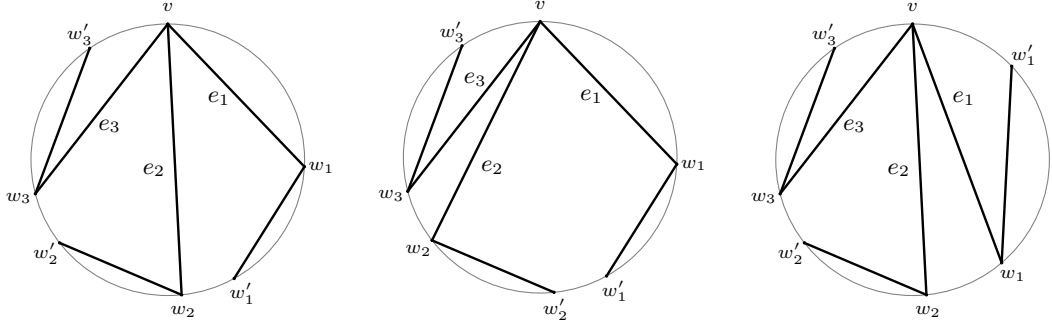
  
\textbf{Case 1:  $w_i'$ is placed to the left of $w_i$ for every $i\in \{1,2,3\}$, as in the left part of Figure~\ref{fig:3configurations}.} In this case, the edges $e_1$ and $e_3$ are of length at least $2^{d-1}-2^{d-2}$, and hence the vertices $w_1,w_1',w_2,w_2',w_3$ are all contained in the interval $[w_1, w_3]$, which contains at most \begin{equation*}
    2^d- (2^{d-1}-2^{d-2}) - (2^{d-1} - 2^{d-2}) +1 = 2^{d-1}+1
\end{equation*}  vertices of $\mathcal{H}_d$. However, the edges $w_1w_1'$ and $w_2w_2'$ are both of length at least $2^{d-1}-2^{d-2}$, and therefore the interval $[w_1,w_2']$ needs to contain at least \begin{equation*}
    2(2^{d-1} - 2^{d-2})+1 = 2^{d-1}+1
\end{equation*}
vertices of $\mathcal{H}_d$, which is impossible since $[w_1,w_2']$ is properly contained in the interval $[w_1,w_3]$. 

\textbf{Case 2: the leaves $w_1',w_3'$ are to the left of $w_1$ and $w_3$, while $w_2'$ is placed to the right of $w_2$, as in the middle part of Figure~\ref{fig:3configurations}.} Then, the vertices $w_1,w_1',w_2',w_2$ are all contained inside the interval $[w_1,w_2]$. As before, the edge $e_1$ is of length at least $2^{d-1}-2^{d-2}$, and $e_2$ is of length at least $2^{d-1}-2^{d-3}$, so the interval $[w_1,w_2]$ contains at most $2^{d-2}+2^{d-3}+1$ vertices of $\mathcal{H}_d$.  On the other hand, the edges $w_1w_1'$ and $w_2w_2'$ are both of length at least $2^{d-1}-2^{d-2}$ and $w_1' \neq w_2'$ and therefore, for the edges $w_1w_1'$ and $w_2w_2'$ to be embedded without crossing, it would be necessary for the interval $[w_1,w_2]$ to contain at least $2^{d-1}+1$ vertices of $\mathcal{H}_d$, contradicting our previous calculations.

\textbf{Case 3: $w_2',w_3'$ is placed to the left of $w_2$ and $w_3$, while $w_1'$ is placed to the right of $w_1$, as shown in the right part of Figure~\ref{fig:3configurations}.} Edges $w_1w_1'$ and $w_3w_3'$ are both of length at least $2^{d-1}-2^{d-2}$, and hence, edges $e_1$ and $e_3$ are both of length at least $2^{d-1}-2^{d-3}$. Vertices $w_1,w_2,w_2',w_3$ all lie inside the interval $[w_1,w_3]$, which contains at most $2^{d-2}+1$ vertices of $\mathcal{H}_d$. Since the edge $w_2w_2'$ is of length at least $2^{d-1} - 2^{d-2} = 2^{d-2}$ and the vertices $w_2,w_2'$ lie in the interval $(w_1,w_3)$, it is necessary that this interval also contains at least $2^{d-2}+1$ vertices of $\mathcal{H}_d$, which is impossible since $(w_1,w_3) \subsetneq [w_1,w_3]$. 

According to the previous discussion, if $G$ is a graph that can be embedded in every drawing of $Q_d$ for sufficiently large $d$, then $G$ must be acyclic and cannot contain $G_0$ as a subgraph; therefore, $G$ must be a forest of caterpillars.  
\end{proof}

\section{Maximum Rectilinear Crossing number of the \texorpdfstring{$\mathbf{d}$}{d}-cube}\label{sec:maxcrossing}

We now define another drawing of $Q_d$, in a similar way to the way we defined $\mathcal{H}_d$. Let $C$ be as before and let $\mathcal{R}_2$ be a drawing of $Q_2$ defined by mapping the vertices of $Q_2$ as follows: $\overline{00} \to 0$, $\overline{01}\to\frac{1}{2}$, $\overline{11}\to \frac{1}{4}$, $\overline{10}\to \frac{3}{4}$. For $d\ge 3$, we obtain $\mathcal{R}_d$ from $\mathcal{R}_{d-1}$ in the same way we obtained $\mathcal{H}_d$ from $\mathcal{H}_{d-1}$; see Figure~\ref{fig:R_d}. As mentioned before, Alpert et al.~\cite{cuberectcrossing}~constructed a drawing weakly isomorphic to $\mathcal{R}_d$ and conjectured that this drawing achieves $CR_{\max}(Q_d)$. We prove the following more general result.

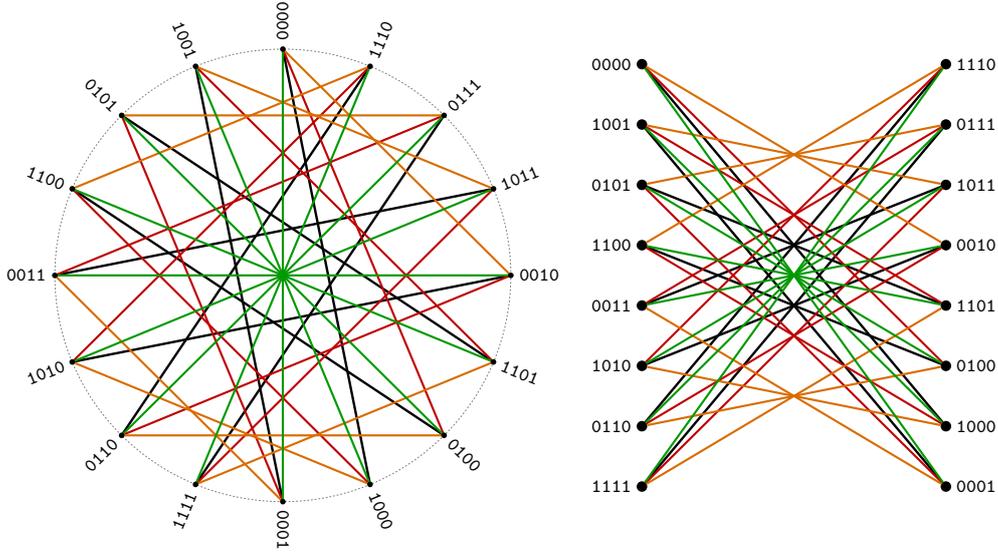
\begin{figure}
  \begin{minipage}{0.45\textwidth}
   \begin{tikzpicture}[line cap=round,line join=round,scale=0.75]

\def\R{4}      
\def\Rl{4.5}   
\def\vdot{1.4pt}

\draw[densely dotted,gray] (0,0) circle (\R);

\foreach \k/\lab in {0/0000,1/1110,2/0111,3/1011,4/0010,5/1101,6/0100,7/1000,8/0001,9/1111,10/0110,11/1010,12/0011,13/1100,14/0101,15/1001}{
  \pgfmathsetmacro{\ang}{90 - 360*\k/16}
  \coordinate (v\lab) at (\ang:\R);
  \coordinate (l\lab) at (\ang:\Rl);
}

\tikzset{
  e1/.style={black, line width=0.9pt},
  e2/.style={red!75!black, line width=0.8pt},
  e3/.style={orange!85!black, line width=0.8pt},
  e4/.style={green!60!black, line width=0.8pt},
}

\draw[e1] (v0000) -- (v1000);
\draw[e1] (v0001) -- (v1001);
\draw[e1] (v0010) -- (v1010);
\draw[e1] (v0011) -- (v1011);
\draw[e1] (v0100) -- (v1100);
\draw[e1] (v0101) -- (v1101);
\draw[e1] (v0110) -- (v1110);
\draw[e1] (v0111) -- (v1111);

\draw[e2] (v0000) -- (v0100);
\draw[e2] (v0001) -- (v0101);
\draw[e2] (v0010) -- (v0110);
\draw[e2] (v0011) -- (v0111);
\draw[e2] (v1000) -- (v1100);
\draw[e2] (v1001) -- (v1101);
\draw[e2] (v1010) -- (v1110);
\draw[e2] (v1011) -- (v1111);

\draw[e3] (v0000) -- (v0010);
\draw[e3] (v0001) -- (v0011);
\draw[e3] (v0100) -- (v0110);
\draw[e3] (v0101) -- (v0111);
\draw[e3] (v1000) -- (v1010);
\draw[e3] (v1001) -- (v1011);
\draw[e3] (v1100) -- (v1110);
\draw[e3] (v1101) -- (v1111);

\draw[e4] (v0000) -- (v0001);
\draw[e4] (v0010) -- (v0011);
\draw[e4] (v0100) -- (v0101);
\draw[e4] (v0110) -- (v0111);
\draw[e4] (v1000) -- (v1001);
\draw[e4] (v1010) -- (v1011);
\draw[e4] (v1100) -- (v1101);
\draw[e4] (v1110) -- (v1111);

\foreach \k/\lab in {0/0000,1/1101,2/0110,3/1010,4/0001,5/1111,6/0100,7/1000,8/0011,9/1110,10/0101,11/1001,12/0010,13/1100,14/0111,15/1011}{

  \path (v\lab) -- (l\lab)
    node[pos=1, sloped, font=\scriptsize\ttfamily, inner sep=1pt] {\lab};

  \fill (v\lab) circle (\vdot);
}
\end{tikzpicture}
\end{minipage}
\hfill
\begin{minipage}{0.45\textwidth}
\begin{tikzpicture}[line cap=round,line join=round, scale=0.8]


\def\vdot{1.4pt}

\tikzset{
  vertex/.style={circle, fill, inner sep=1.4pt} 
}

\foreach \k/\lab in {8/1111,9/0110,10/1010,11/0011,12/1100,13/0101,14/1001,15/0000}{
 \node[vertex] (v\lab) at (1,\k){};
 \node[font=\scriptsize\ttfamily, inner sep=1pt] (l\lab) at (0.5,\k) {\lab};
   
}
\foreach \k/\lab in {15/1110,14/0111,13/1011,12/0010,11/1101,10/0100,9/1000,8/0001}{
 \node[vertex] (v\lab) at (6,\k){};
 \node[font=\scriptsize\ttfamily, inner sep=1pt] (l\lab) at (6.5,\k) {\lab};

  \fill (v\lab) circle (\vdot);
}

\tikzset{
  e1/.style={black, line width=0.9pt},
  e2/.style={red!75!black, line width=0.8pt},
  e3/.style={orange!85!black, line width=0.8pt},
  e4/.style={green!60!black, line width=0.8pt},
}

\draw[e1] (v0000) -- (v1000);
\draw[e1] (v0001) -- (v1001);
\draw[e1] (v0010) -- (v1010);
\draw[e1] (v0011) -- (v1011);
\draw[e1] (v0100) -- (v1100);
\draw[e1] (v0101) -- (v1101);
\draw[e1] (v0110) -- (v1110);
\draw[e1] (v0111) -- (v1111);

\draw[e2] (v0000) -- (v0100);
\draw[e2] (v0001) -- (v0101);
\draw[e2] (v0010) -- (v0110);
\draw[e2] (v0011) -- (v0111);
\draw[e2] (v1000) -- (v1100);
\draw[e2] (v1001) -- (v1101);
\draw[e2] (v1010) -- (v1110);
\draw[e2] (v1011) -- (v1111);

\draw[e3] (v0000) -- (v0010);
\draw[e3] (v0001) -- (v0011);
\draw[e3] (v0100) -- (v0110);
\draw[e3] (v0101) -- (v0111);
\draw[e3] (v1000) -- (v1010);
\draw[e3] (v1001) -- (v1011);
\draw[e3] (v1100) -- (v1110);
\draw[e3] (v1101) -- (v1111);

\draw[e4] (v0000) -- (v0001);
\draw[e4] (v0010) -- (v0011);
\draw[e4] (v0100) -- (v0101);
\draw[e4] (v0110) -- (v0111);
\draw[e4] (v1000) -- (v1001);
\draw[e4] (v1010) -- (v1011);
\draw[e4] (v1100) -- (v1101);
\draw[e4] (v1110) -- (v1111);

\end{tikzpicture}
\end{minipage}
    \caption{Left: $\mathcal{R}_4$; Right: The drawing from the construction of Alpert et al. for $Q_4$.}
    \label{fig:R_d}

\end{figure}


\lengthregcrossings*

\begin{proof}
    Let, $d\ge 2$ and let $\mathcal{D}$ be a length-regular drawing of $G$ with length profile 
    $(\ell_1,\ell_2,\dots,\\\ell_d)$. Note that if $\ell_i=\ell_{i+1}$ for some $i$ (i.e., if the drawing contains all possible edges of this length), we consider half of the edges of this length as edges of length $\ell_i$ and the other half as edges of length $\ell_{i+1}$ in a way that each vertex is incident to one of each of them.
    An edge $e=vw$ of length $\ell_i$ and an edge $f=xy$ of length $\geq \ell_i$, cross precisely if one of $x,y$ lies in the interval $(v,w)$ and the other lies in the interval $(w,v)$.
    As the drawing is length-regular, for each of the $\ell_i-1$ vertices between $v$ and $w$ there is an edge for each length, in particular for each $\ell_j$ with $j<i$. 
    Thus, we count in total $(\ell_i-1)(i-1)$ crossings of $e$ with edges of length $\ell_1,\ldots,\ell_{i-1}$ and $(\ell_i-1)$ crossings of $e$ with edges of length $\ell_i$.
    Adding the crossings of all $\frac{|V(G)|}{2}$ edges of length $\ell_i$ with edges of length $\geq\ell_i$, we obtain $\frac{|V(G)|}{2}(\ell_i-1)(i-\frac{1}{2})$ crossings, where the $i-\frac{1}{2}$ appears since crossings between two edges of length $\ell_i$ are counted twice. Summing over all lengths, we see that the total number of crossings in $\mathcal{D}$ is $\frac{|V(G)|}{2}\sum_{i=1}^d(\ell_i-1)(i-\frac{1}{2})$.
\end{proof}

\begin{corollary}
    As the drawing $\mathcal{R}_d$ is length-regular with lengths $\ell_1=2^{d-1}$ and $\ell_i=2^{d-1}-2^{i-2}$ for $2\le i\le d$, 
    the maximum rectilinear crossing number CR$_{\max}(Q_d)$ is at least~$2^{d-2}(2^{d-1}(d^2-2d+3)-d^2-1)$. 
\end{corollary}

We note that Theorem~\ref{thm:lengthregcrossings} allows us to recover simpler proofs of other results about maximum rectilinear crossing numbers. For example, \emph{star-like} drawings of $d$-regular graphs presented in~\cite{kregrectcross} are also length-regular. These drawings are known to maximize the maximum rectilinear crossing number taken over all $d$-regular graphs with a given number of vertices.  

\section{Plane paths in rectilinear and simple drawings of the $d$-cube}\label{sec:nonconvexdrawings}

\fourpathrectilinear*

\begin{proof}
Consider an arbitrary rectilinear drawing $\mathcal{D}$ of $Q_3$. We consider two cases. \\
\textbf{Case 1: $\mathcal{D}$ contains a plane $4$-cycle.} Let the vertices of this plane $4$-cycle be $a,b,c,d$, we need to consider two subcases, depending on whether $a,b,c,d$ form a convex or non-convex $4$-gon. \\
\textit{Subcase 1.1: $a,b,c,d$ form a convex $4$-gon.} Let $a',b',c',d'$ be the remaining vertices of $Q_3$ so that $xx'$ is an edge for every $x\in \{a,b,c,d\}$ . For each vertex $x\in \{a,b,c,d\}$ there are exactly two plane paths $p_x,q_x$ in the $4$-gon $\{a,b,c,d\}$ that have $x$ as an endpoint. For example $p_a= d,c,b,a$ and $q_a=b,c,d,a$. Now consider the paths $p_a'=d,c,b,a,a'$ and $q_a'=b,c,d,a,a'$. If any of these paths is plane, we are done, so we can assume that the edge $aa'$ crosses one of the edges $bc$ or $cd$. Therefore, $a'$ must lie in the section of the plane defined by the angle $\angle dab$ and outside of the convex hull of $\{a,b,c,d\}$, see Figure~\ref{fig:4pathgeometriccase11}, we call this a \emph{plausible region} for $a'$. Similarly, we can conclude the plausible regions for $b',c',d'$. Now, if the edge $a'b'$ does not cross the edge $cb$, we have a plane path of length $4$ given by $b',a',a,b,c$. Otherwise, if $b'a'$ crosses $bc$, note that it is impossible for $a'd'$ to cross $dc$ and therefore $d',a',a,d,c$ is a plane path of length $4$ in $\mathcal{D}$. 

\begin{figure}[ht]
    \centering
    \includegraphics[page=1]{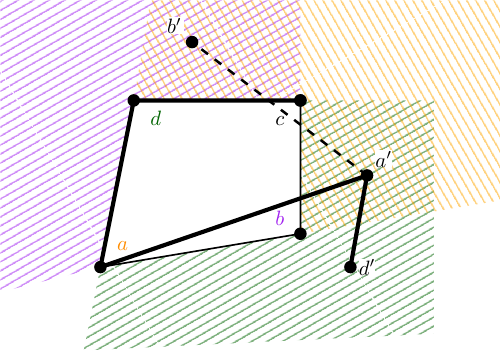}
    \caption{Subcase 1.1 in the proof of Proposition~\ref{prop:4pathrectilinear}. The regions in which vertices $a',c',d'$ can be placed are shaded, plane path $d',a',a,d,c$ is drawn in bold.}
    \label{fig:4pathgeometriccase11}
\end{figure}

\textit{Subcase 1.2: $a,b,c,d$ form a non-convex $4$-gon.} Assume that $a,b,c,d$ are placed as in Figure~\ref{fig:4pathgeometriccase12}, i.e., that $c$ is inside the triangle $\triangle abd$. Let $a',b',c',d'$ and $p_x,q_x$ be defined as in the previous subcase. By the same arguments as before, we can find plausible regions in which $a',b',d'$ can lie, which can be seen in Figure~\ref{fig:4pathgeometriccase12}. Again, if the edge $b'a'$ does not cross the edge $bc$, the path $b',a',a,b,c$ is plane and we are done. Otherwise, it is impossible for $a'b'$ to cross $bc$ and $a'd'$ to cross $dc$ so the path $d',a',a,d,c$ is a plane path of length $4$ in $\mathcal{D}$. 

\begin{figure}[ht]
    \centering
    \includegraphics[page=2]{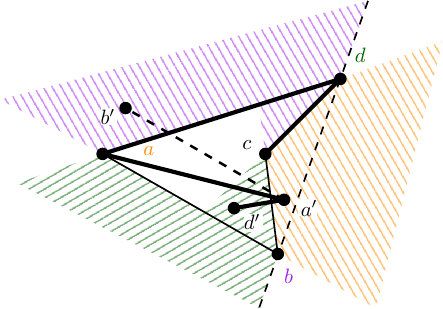}
    \caption{Subcase 1.2 in the proof of Proposition~\ref{prop:4pathrectilinear}.}
    \label{fig:4pathgeometriccase12}
\end{figure}

\textbf{Case 2: There is no plane $4$-cycle in $D$.} Let $a,b,c,d$ be an arbitrary $4$-cycle in $D$ and assume that the edges $ab$ and $cd$ cross. Let $a',b',c',d'$ be defined as in the previous cases. Additionally, since $a,b,c,d$ is not plane, for each $x\in \{a,b,c,d\}$ there is a unique plane path of length $3$ with $x$ as an endpoint in the $4$-cycle $a,b,c,d$, denote this path $p_x$. Using similar arguments as in previous cases, we can determine plausible regions for $a',b',c',d'$, which can be seen in Figure~\ref{fig:4pathgeometriccase2regions}.  Furthermore, if the edge $c'd'$ does not cross both of the edges $ad,bc$, then either $a,d,d',c',c$ or $d,d',c',c,b$ is a plane path of length $4$, so we assume that $c'd'$ crosses both of these edges, similarly, we assume that $a'b'$ crosses both of these edges.
Consider the $4$-cycle $c',d',d,c$. We know that this $4$-cycle is not 
plane, so either $c'd'$ crosses $cd$ or $c'c$ crosses $d'd$. 

\begin{figure}[ht]
    \centering
    \includegraphics[page=3]{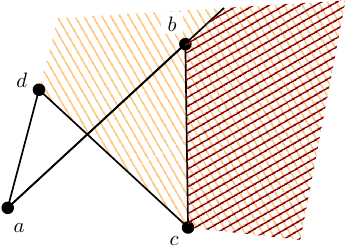}
    \caption{Setup of case 2 in the proof of Proposition~\ref{prop:4pathrectilinear}. Before considering the fact that $a'b'$ must cross both edges $ad$ and $bc$, we can conclude that $a'$ can be placed into the orange region. With additional information, $a'$ can only be placed in the dark red shaded region.}
    \label{fig:4pathgeometriccase2regions}
\end{figure}

\textit{Subcase 2.1: $c'c$ crosses $d'd$.} In this case, $c'd'$ does not cross either $ad$ or $bc$, which contradicts our assumptions.\\
\textit{Subcase 2.2: $c'd'$ crosses $cd$.}  Now, $a$ and $b$ lie in different halfplanes determined by line $\ell$ through $c',d'$. If $a'$ and $b'$ are both in the same halfplane determined by $\ell$, then the $4$-cycle $a',b',c',d'$ is plane, contradicting the assumption of this case. So, we can assume that $a'$ lies in the same halfplane determined by $\ell$ as $b$ and $b'$ lies in the same halfplane determined by $\ell$ as $a$. Recall that we assume that $a'b'$ crosses both $ad$ and $bc$. Therefore, we are in the situation as in Figure~\ref{fig:4pathgeometriccase22}, so the path $c',b',b,a,a'$ is a plane path of length $4$. \qedhere
 
\begin{figure}[ht]
    \centering
    \includegraphics[page=4]{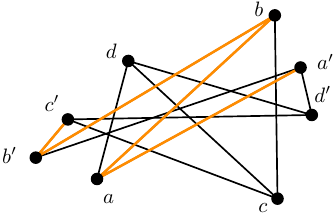}
    \caption{Subcase 2.2 in the proof of Proposition~\ref{prop:4pathrectilinear}, the plane path of length $4$ $c',b',b,a,a'$ is drawn in orange.}
    \label{fig:4pathgeometriccase22}
\end{figure}

\end{proof}

Finally, we construct a simple drawing of $Q_3$ with no plane path of length $4$. Our construction can be seen in Figure~\ref{fig:Q3} below. 

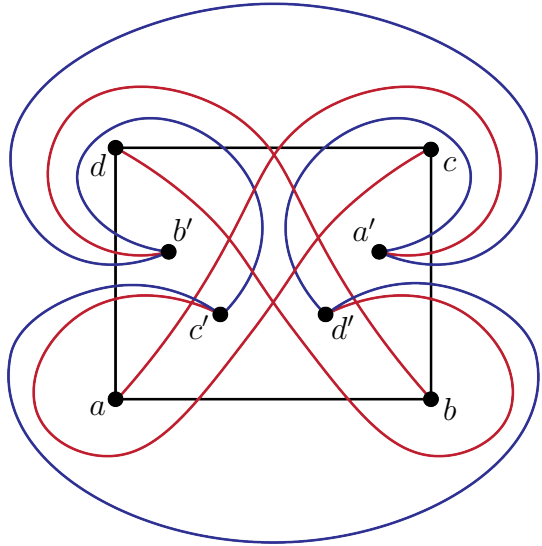
\begin{figure}
    \centering
   \begin{tikzpicture}
    \node[anchor=south west, inner sep=0] (img) at (0,0)
        {\includegraphics[width=7cm]{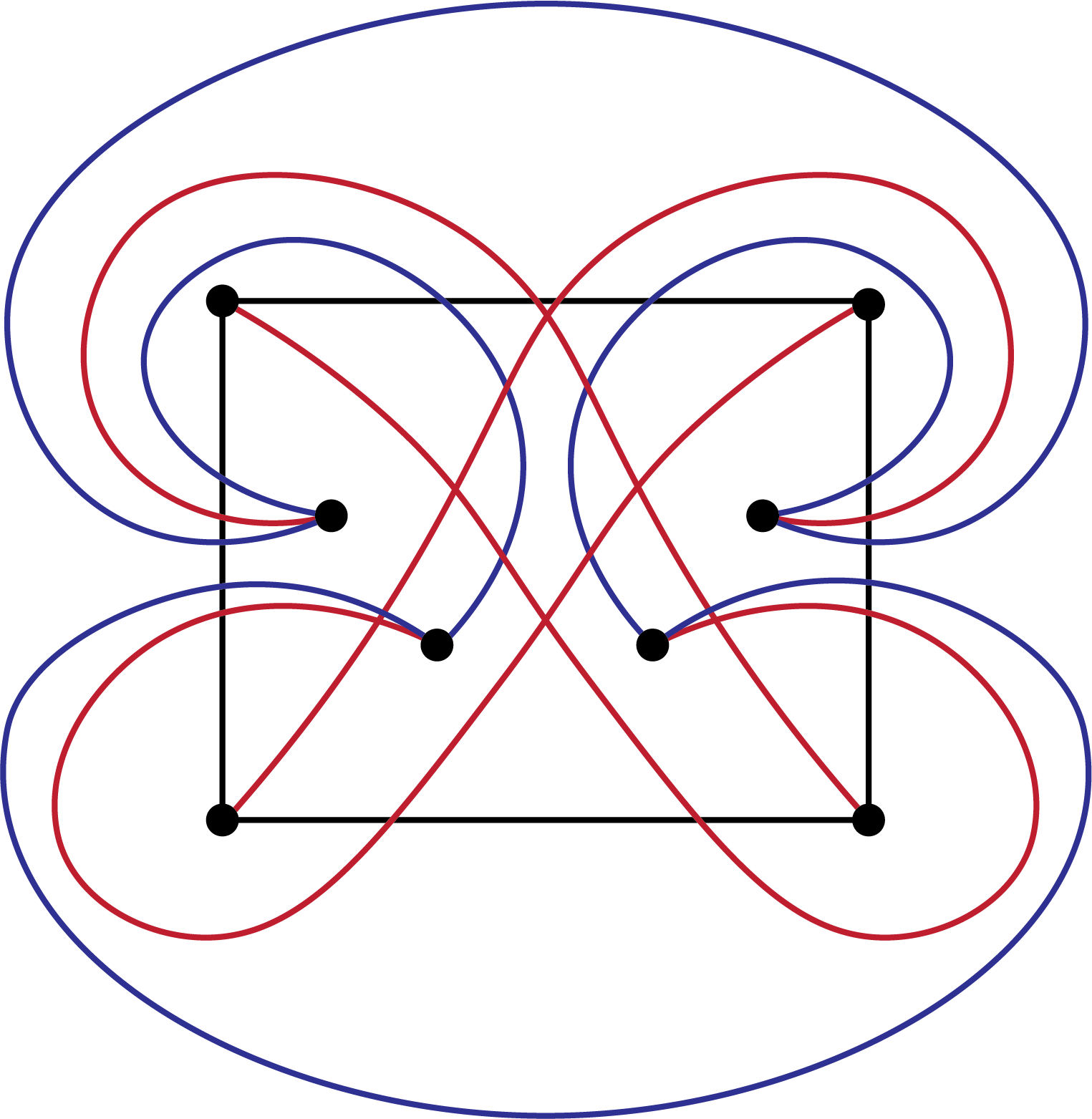}};
    \begin{scope}[x={(img.south east)}, y={(img.north west)}]

        \node at (0.17,0.7) {\Large $d$};      
        \node at (0.83,0.25) {\Large $b$};       
        \node at (0.17,0.25) {\Large $a$};       
        \node at (0.83,0.7) {\Large $c$};        
         \node at (0.63,0.4) {\Large $d'$};         
        \node at (0.67,0.58) {\Large $a'$};         
        \node at (0.33,0.58) {\Large $b'$};         
        \node at (0.36,0.4) {\Large $c'$};       
    \end{scope}
    
\end{tikzpicture}
    \caption{Simple drawing of $Q_3$ with no plane path of length $4$.}
    \label{fig:Q3}
\end{figure}

\QthreeNoPlanePath* 

\begin{proof}
    Let $\mathcal{D}$ be the drawing in Figure~\ref{fig:Q3}, we refer to edges by their color (red/blue/black) and to vertices by their label in the figure. We will show that $\mathcal{D}$ does not contain a plane path of length $4$.  Suppose for contradiction that there is a plane path $p$ of length $4$ in $\mathcal{D}$. Observe that every red edge in $\mathcal{D}$ crosses every other red edge and therefore $p$ can contain only one red edge. We now differentiate between $4$ cases based on the order of edges within~$p$. \\
    \textbf{Case 1: $p$ consists of three consecutive blue edges and one red edge.} Observe that each red edge crosses the blue cycle at least once. For example, the edge $bb'$ crosses the edge $a'd'$. Now, if $p$ contains the edge $bb'$, by the assumptions of the case we are in, we can see that $p=a,d',c',b',b$ and is hence not plane. Similar analysis holds for the other three red edges. \\
    \textbf{Case 2: $p$ consists of two consecutive blue edges, one red edge and one black edge.} Assume that $aa'$ is the single red edge in $p$. Since it is preceded by two blue edges, we know that $p=c',d',a',a,x$ where $x\in \{b,d\}$. However, both the edges $ab$ and $ad$ cross one of the edges $c'd'$ or $d'a'$, hence $p$ is not plane. Similar analysis holds for the other red edges. \\
    \textbf{Case 3: $p$ consists of one blue edge, one red edge and two consecutive black edges.} Note that every red edge crosses exactly two black edges, so the only black edges not crossed by the red edges share a vertex with it. Therefore, it is impossible for a red edge to be preceded by more than one black edge in any plane path. \\
    \textbf{Case 4: $p$ consists of one red edge and three consecutive black edges.} Same as Case 3. 
\end{proof}

\section{Discussion and Open Problems}

To the best of our knowledge, we were the first to investigate existence of plane substructures in drawings of $Q_d$. While we gave partial or full answer to some questions, many still remain open. We highlight the ones that we consider the most interesting here. 

Firstly, our study was mostly concerned with convex-geometric drawings of $Q_d$, and we did not investigate general rectilinear or simple drawings of $Q_d$ in much detail. Investigating the problems in these settings may pose some interesting challenges. We attempted to find a counterpart to Theorem~\ref{cor_longplanepaths} for general rectilinear drawings. However, the bound that we obtain in Proposition~\ref{prop:4pathrectilinear} is considerably worse than that obtained from Theorem~\ref{cor_longplanepaths} for convex-geometric drawings. 
Note that the upper bound on the length of the longest plane path from Theorem~\ref{thm:drawingswithnolongpaths} is still valid for general rectilinear drawings, so the gap between the lower and upper bounds is very large. It would be interesting to improve any of these bounds. 

For simple drawings, we know even less. We constructed a drawing of $Q_3$ that does not contain a plane path of length more than $3$. Note that a path of length $3$ can always be found already in $Q_2$. It is possible (but it would be surprising) that our construction can be extended to higher dimensions. This would be particularly interesting since determining the maximum number of edges that a simple drawing can have without containing a plane path of given length was investigated recently by Keszegh et al.~\cite{keszegh}, so such constructions would be interesting from this perspective as well.  

Another intriguing problem would be to close the gap between the bounds of Theorems~\ref{cor_longplanepaths} and~\ref{thm-planesubgraphrotateddrawing}. Currently, we know that every convex-geometric drawing of $Q_d$ contains a plane path of length at least $d-1$ and that some drawings ($\mathcal{H}_d$) do not contain plane paths of length larger than $2d-3$. We are not sure which (if any) of these bounds is tight. 

Lastly, we mention the conjecture of Alpert et al.~\cite{cuberectcrossing}, which says that the maximum crossing number CR$_{\max}$ of $Q_d$ is obtained by the drawing $\mathcal{R}_d$. Since this conjecture seems to be hard to solve, we think that verifying it for the more restricted classes of convex-geometric or even length-regular convex-geometric drawings of $Q_d$ might be more approachable. A more ambitious direction might be to consider the maximum crossing number of $Q_d$ for general simple drawings. 

\bibstyle{plainurl}

\bibliography{bibliography}

@inproceedings {cuberectcrossing,
    AUTHOR = {Alpert, Matthew and Feder, Elie and Harborth, Heiko and Klein,
              Sheldon},
     TITLE = {The maximum rectilinear crossing number of the {$n$}
              dimensional cube graph},
 BOOKTITLE = {Proceedings of the {F}ortieth {S}outheastern {I}nternational
              {C}onference on {C}ombinatorics, {G}raph {T}heory and
              {C}omputing},
   JOURNAL = {Congr. Numer.},
  FJOURNAL = {Congressus Numerantium. A Conference Journal on Numerical
              Themes},
    VOLUME = {195},
      YEAR = {2009},
     PAGES = {147--158},
      ISSN = {0384-9864},
   MRCLASS = {05C10},
  MRNUMBER = {2584292},
}

@article {completerectcrossing,
    AUTHOR = {Harborth, Heiko and Th\"{u}rmann, Christian},
     TITLE = {Maximum rectilinear crossing number in drawings of the
              complete graph with a given convex hull},
   JOURNAL = {Congr. Numer.},
  FJOURNAL = {Congressus Numerantium. A Conference Journal on Numerical
              Themes},
    VOLUME = {221},
      YEAR = {2014},
     PAGES = {121--127},
      ISSN = {0384-9864},
   MRCLASS = {05C35 (05C62)},
  MRNUMBER = {3329011},
}

@article {cyclesrectcross,
    AUTHOR = {Bode, Jens-P. and Feder, Elie and Harborth, Heiko and
              Horowitz, David and Lichter, Tamar},
     TITLE = {Extremal values of the maximum rectilinear crossing number of
              cycles with diagonals},
   JOURNAL = {Congr. Numer.},
  FJOURNAL = {Congressus Numerantium. A Conference Journal on Numerical
              Themes},
    VOLUME = {220},
      YEAR = {2014},
     PAGES = {33--48},
      ISSN = {0384-9864},
   MRCLASS = {05C35 (05C62)},
  MRNUMBER = {3308551},
}

@article{kregrectcross,
  title={The Maximum of the Maximum Rectilinear Crossing Numbers of $d$-Regular Graphs of Order $n$},
  author={Matthew Alpert and Elie Feder and Heiko Harborth},
  journal={Electron. J. Comb.},
  year={2008},
  volume={16},
  url={https://api.semanticscholar.org/CorpusID:9548696}
}

@article{Aichholzer2024,
  title = {Twisted Ways to Find Plane Structures in Simple Drawings of Complete Graphs},
  volume = {71},
  ISSN = {1432-0444},
  DOI = {10.1007/s00454-023-00610-0},
  number = {1},
  journal = {Discrete \& Computational Geometry},
  publisher = {Springer Science and Business Media LLC},
  author = {Aichholzer,  Oswin and García,  Alfredo and Tejel,  Javier and Vogtenhuber,  Birgit and Weinberger,  Alexandra},
  year = {2024},
  month = jan,
  pages = {40–66}
}

@inproceedings{perles,
    author = {Moser, W. and Pach, J.},
    title = {Geometric graphs},
    booktitle = {New Trends in Discrete and
Computational Geometry (J. Pach, ed.)},
    publisher = {Springer, New York},
    year = {1993}
}

@inproceedings{Aichholzer2023,
  doi = {10.4230/LIPICS.SOCG.2023.6},
  url = {https://drops.dagstuhl.de/entities/document/10.4230/LIPIcs.SoCG.2023.6},
  author = {Aichholzer,  Oswin and Chiu,  Man-Kwun and Hoang,  Hung P. and Hoffmann,  Michael and Kynčl,  Jan and Maus,  Yannic and Vogtenhuber,  Birgit and Weinberger,  Alexandra},
  language = {en},
  title = {Drawings of Complete Multipartite Graphs up to Triangle Flips},
  journal = {LIPIcs,  Volume 258,  SoCG 2023},
  volume = {258},
  pages = {6:1--6:16},
  publisher = {Schloss Dagstuhl – Leibniz-Zentrum f\"{u}r Informatik},
  year = {2023},
  copyright = {Creative Commons Attribution 4.0 International license}
}

@article{CardinalFelsner2018,
  doi = {10.20382/JOCG.V9I1A7},
  url = {https://jocg.org/index.php/jocg/article/view/3053},
  author = {Cardinal,  Jean and Felsner,  Stefan},
  language = {en},
  title = {Topological drawings of complete bipartite graphs},
  journal = {Journal of Computational Geometry},
  pages = {Vol. 9 No. 1 (2018)},
  publisher = {Journal of Computational Geometry},
  year = {2018}
}

@article{deKlerk2014,
  title = {Book drawings of complete bipartite graphs},
  volume = {167},
  ISSN = {0166-218X},
  url = {http://dx.doi.org/10.1016/j.dam.2013.11.001},
  DOI = {10.1016/j.dam.2013.11.001},
  journal = {Discrete Applied Mathematics},
  publisher = {Elsevier BV},
  author = {de Klerk,  Etienne and Pasechnik,  Dmitrii V. and Salazar,  Gelasio},
  year = {2014},
  month = apr,
  pages = {80–93}
}

@inbook{Lenhart2023,
  title = {Mutual Witness Gabriel Drawings of Complete Bipartite Graphs},
  ISBN = {9783031222030},
  ISSN = {1611-3349},
  url = {http://dx.doi.org/10.1007/978-3-031-22203-0_3},
  DOI = {10.1007/978-3-031-22203-0_3},
  booktitle = {Graph Drawing and Network Visualization},
  publisher = {Springer International Publishing},
  author = {Lenhart,  William J. and Liotta,  Giuseppe},
  year = {2023},
  pages = {25–39}
}

@article{Kynl2009,
  title = {Enumeration of simple complete topological graphs},
  volume = {30},
  ISSN = {0195-6698},
  url = {http://dx.doi.org/10.1016/j.ejc.2009.03.005},
  DOI = {10.1016/j.ejc.2009.03.005},
  number = {7},
  journal = {European Journal of Combinatorics},
  publisher = {Elsevier BV},
  author = {Kynčl,  Jan},
  year = {2009},
  month = oct,
  pages = {1676–1685}
}

@article{Arroyo2021,
  title = {Drawings of complete graphs in the projective plane},
  volume = {97},
  ISSN = {1097-0118},
  url = {http://dx.doi.org/10.1002/jgt.22665},
  DOI = {10.1002/jgt.22665},
  number = {3},
  journal = {Journal of Graph Theory},
  publisher = {Wiley},
  author = {Arroyo,  Alan and McQuillan,  Dan and Richter,  R. Bruce and Salazar,  Gelasio and Sullivan,  Matthew},
  year = {2021},
  month = mar,
  pages = {426–440}
}

@article{Balko2025,
  title = {Faces in rectilinear drawings of complete graphs},
  volume = {130},
  ISSN = {0195-6698},
  url = {http://dx.doi.org/10.1016/j.ejc.2025.104217},
  DOI = {10.1016/j.ejc.2025.104217},
  journal = {European Journal of Combinatorics},
  publisher = {Elsevier BV},
  author = {Balko,  Martin and Br\"{o}tzner,  Anna and Klute,  Fabian and Tkadlec,  Josef},
  year = {2025},
  month = dec,
  pages = {104217}
}

@article{Arroyo2021a,
  title = {Extending Drawings of Complete Graphs into Arrangements of Pseudocircles},
  volume = {35},
  ISSN = {1095-7146},
  url = {http://dx.doi.org/10.1137/20M1313234},
  DOI = {10.1137/20m1313234},
  number = {2},
  journal = {SIAM Journal on Discrete Mathematics},
  publisher = {Society for Industrial \& Applied Mathematics (SIAM)},
  author = {Arroyo,  Alan and Richter,  R. Bruce and Sunohara,  Matthew},
  year = {2021},
  month = jan,
  pages = {1050–1076}
}

@article{Gioan2022,
  title = {Complete Graph Drawings up to Triangle Mutations},
  volume = {67},
  ISSN = {1432-0444},
  url = {http://dx.doi.org/10.1007/s00454-021-00339-8},
  DOI = {10.1007/s00454-021-00339-8},
  number = {4},
  journal = {Discrete \& Computational Geometry},
  publisher = {Springer Science and Business Media LLC},
  author = {Gioan,  Emeric},
  year = {2022},
  month = mar,
  pages = {985–1022}
}

@article{brego2013,
  title = {The 2-Page Crossing Number of ${K}_{n}$ },
  volume = {49},
  ISSN = {1432-0444},
  url = {http://dx.doi.org/10.1007/s00454-013-9514-0},
  DOI = {10.1007/s00454-013-9514-0},
  number = {4},
  journal = {Discrete \& Computational Geometry},
  publisher = {Springer Science and Business Media LLC},
  author = {Ábrego,  Bernardo M. and Aichholzer,  Oswin and Fernández-Merchant,  Silvia and Ramos,  Pedro and Salazar,  Gelasio},
  year = {2013},
  month = jun,
  pages = {747–777}
}

@article {MR1775301,
    AUTHOR = {Faria, Luerbio and de Figueiredo, Celina Miraglia Herrera},
     TITLE = {On {E}ggleton and {G}uy's conjectured upper bound for the
              crossing number of the {$n$}-cube},
   JOURNAL = {Math. Slovaca},
  FJOURNAL = {Mathematica Slovaca},
    VOLUME = {50},
      YEAR = {2000},
    NUMBER = {3},
     PAGES = {271--287},
      ISSN = {0139-9918,1337-2211},
   MRCLASS = {05C10 (57M15)},
  MRNUMBER = {1775301},
MRREVIEWER = {L.\ Yu.\ Glebski\u{\i}},
}

@article{Kainen1972,
  title = {A lower bound for crossing numbers of graphs with applications to ${K}_n$, ${K}_{p,q}$, and ${Q}(d)$},
  volume = {12},
  ISSN = {0095-8956},
  url = {http://dx.doi.org/10.1016/0095-8956(72)90042-1},
  DOI = {10.1016/0095-8956(72)90042-1},
  number = {3},
  journal = {Journal of Combinatorial Theory,  Series B},
  publisher = {Elsevier BV},
  author = {Kainen,  Paul C},
  year = {1972},
  month = jun,
  pages = {287–298}
}

@inbook{Skora1992,
  title = {On the crossing number of the hypercube and the cube connected cycles},
  ISBN = {9783540467359},
  ISSN = {1611-3349},
  url = {http://dx.doi.org/10.1007/3-540-55121-2_21},
  DOI = {10.1007/3-540-55121-2_21},
  booktitle = {Graph-Theoretic Concepts in Computer Science},
  publisher = {Springer Berlin Heidelberg},
  author = {Sýkora,  Ondrej and Vrťo,  Imrich},
  year = {1992},
  pages = {214–218}
}

@inbook{Kindermann2023,
  title = {Three Edge-Disjoint Plane Spanning Paths in a Point Set},
  ISBN = {9783031492723},
  ISSN = {1611-3349},
  url = {http://dx.doi.org/10.1007/978-3-031-49272-3_22},
  DOI = {10.1007/978-3-031-49272-3_22},
  booktitle = {Graph Drawing and Network Visualization},
  publisher = {Springer Nature Switzerland},
  author = {Kindermann,  P. and Kratochvíl,  J. and Liotta,  G. and Valtr,  P.},
  year = {2023},
  pages = {323–338}
}

@misc{arxiv,
  doi = {10.48550/ARXIV.2511.22526},
  url = {https://arxiv.org/abs/2511.22526},
  author = {Antić,  Todor and Džuklevski,  Aleksa and Fiala,  Jiří and Kratochvíl,  Jan and Liotta,  Giuseppe and Saghafian,  Morteza and Saumell,  Maria and Zink,  Johannes},
  title = {Edge-Constrained Hamiltonian Paths on a Point Set},
  publisher = {arXiv},
  year = {2025},
  copyright = {arXiv.org perpetual,  non-exclusive license}
}

@article{Kynl2008,
  title = {Long alternating paths in bicolored point sets},
  volume = {308},
  ISSN = {0012-365X},
  url = {http://dx.doi.org/10.1016/j.disc.2007.08.013},
  DOI = {10.1016/j.disc.2007.08.013},
  number = {19},
  journal = {Discrete Mathematics},
  publisher = {Elsevier BV},
  author = {Kynčl,  Jan and Pach,  János and Tóth,  Géza},
  year = {2008},
  month = oct,
  pages = {4315–4321}
}

@inproceedings{mulzer2020,
  doi = {10.4230/LIPICS.SOCG.2020.57},
  url = {https://drops.dagstuhl.de/entities/document/10.4230/LIPIcs.SoCG.2020.57},
  author = {Mulzer,  Wolfgang and Valtr,  Pavel},
  language = {en},
  title = {Long Alternating Paths Exist},
  journal = {LIPIcs,  Volume 164,  SoCG 2020},
  volume = {164},
  pages = {57:1--57:16},
  publisher = {Schloss Dagstuhl – Leibniz-Zentrum f\"{u}r Informatik},
  year = {2020},
  copyright = {Creative Commons Attribution 3.0 Unported license}
}

@misc{soukup2024,
  doi = {10.48550/ARXIV.2404.06105},
  url = {https://arxiv.org/abs/2404.06105},
  author = {Soukup,  Jan},
  title = {Bicolored point sets admitting non-crossing alternating Hamiltonian paths},
  publisher = {arXiv},
  year = {2024},
  copyright = {arXiv.org perpetual,  non-exclusive license}
}

@inbook{Cibulka2009,
  title = {Hamiltonian Alternating Paths on Bicolored Double-Chains},
  ISBN = {9783642002199},
  ISSN = {1611-3349},
  url = {http://dx.doi.org/10.1007/978-3-642-00219-9_18},
  DOI = {10.1007/978-3-642-00219-9_18},
  booktitle = {Graph Drawing},
  publisher = {Springer Berlin Heidelberg},
  author = {Cibulka,  Josef and Kynčl,  Jan and Mészáros,  Viola and Stolař,  Rudolf and Valtr,  Pavel},
  year = {2009},
  pages = {181–192}
}

@article{Cska2022,
  title = {Upper bounds for the necklace folding problems},
  volume = {157},
  ISSN = {0095-8956},
  url = {http://dx.doi.org/10.1016/j.jctb.2022.05.012},
  DOI = {10.1016/j.jctb.2022.05.012},
  journal = {Journal of Combinatorial Theory,  Series B},
  publisher = {Elsevier BV},
  author = {Csóka,  Endre and Blázsik,  Zoltán L. and Király,  Zoltán and Lenger,  Dániel},
  year = {2022},
  month = nov,
  pages = {123–143}
}

@misc{Ricci2025,
  doi = {10.48550/ARXIV.2506.20421},
  url = {https://arxiv.org/abs/2506.20421},
  author = {Ricci,  Marco and Rollin,  Jonathan and Schulz,  André and Weinberger,  Alexandra},
  title = {On plane cycles in geometric multipartite graphs},
  publisher = {arXiv},
  year = {2025},
  copyright = {Creative Commons Attribution 4.0 International}
}

@inbook{Suk2023,
  title = {Unavoidable Patterns in Complete Simple Topological Graphs},
  ISBN = {9783031222030},
  ISSN = {1611-3349},
  url = {http://dx.doi.org/10.1007/978-3-031-22203-0_1},
  DOI = {10.1007/978-3-031-22203-0_1},
  booktitle = {Graph Drawing and Network Visualization},
  publisher = {Springer International Publishing},
  author = {Suk,  Andrew and Zeng,  Ji},
  year = {2023},
  pages = {3–15}
}

@inbook{Pach2005,
  title = {Disjoint Edges in Topological Graphs},
  ISBN = {9783540305408},
  ISSN = {1611-3349},
  url = {http://dx.doi.org/10.1007/978-3-540-30540-8_15},
  DOI = {10.1007/978-3-540-30540-8_15},
  booktitle = {Combinatorial Geometry and Graph Theory},
  publisher = {Springer Berlin Heidelberg},
  author = {Pach,  János and Tóth,  Géza},
  year = {2005},
  pages = {133–140}
}

@article{GarcaOlaverri2021,
  title = {On plane subgraphs of complete topological drawings},
  volume = {20},
  ISSN = {1855-3966},
  url = {http://dx.doi.org/10.26493/1855-3974.2226.e93},
  DOI = {10.26493/1855-3974.2226.e93},
  number = {1},
  journal = {Ars Mathematica Contemporanea},
  publisher = {University of Primorska Press},
  author = {García Olaverri,  Alfredo and Tejel Altarriba,  Javier and Pilz,  Alexander},
  year = {2021},
  month = aug,
  pages = {69–87}
}

@article{Bergold2025,
  title = {Plane Hamiltonian Cycles in Convex Drawings},
  ISSN = {1432-0444},
  url = {http://dx.doi.org/10.1007/s00454-025-00752-3},
  DOI = {10.1007/s00454-025-00752-3},
  journal = {Discrete \& Computational Geometry},
  publisher = {Springer Science and Business Media LLC},
  author = {Bergold,  Helena and Felsner,  Stefan and M. Reddy,  Meghana and Orthaber,  Joachim and Scheucher,  Manfred},
  year = {2025},
  month = jul 
}

@article{Fulek2014,
  title = {Estimating the Number of Disjoint Edges in Simple Topological Graphs via Cylindrical Drawings},
  volume = {28},
  ISSN = {1095-7146},
  url = {http://dx.doi.org/10.1137/130925554},
  DOI = {10.1137/130925554},
  number = {1},
  journal = {SIAM Journal on Discrete Mathematics},
  publisher = {Society for Industrial \& Applied Mathematics (SIAM)},
  author = {Fulek,  Radoslav},
  year = {2014},
  month = jan,
  pages = {116–121}
}

@inbook{Garca2023,
  title = {Empty Triangles in Generalized Twisted Drawings of $K_n$},
  ISBN = {9783031222030},
  ISSN = {1611-3349},
  url = {http://dx.doi.org/10.1007/978-3-031-22203-0_4},
  DOI = {10.1007/978-3-031-22203-0_4},
  booktitle = {Graph Drawing and Network Visualization},
  publisher = {Springer International Publishing},
  author = {García,  Alfredo and Tejel,  Javier and Vogtenhuber,  Birgit and Weinberger,  Alexandra},
  year = {2023},
  pages = {40–48}
}

@misc{keszegh,
  doi = {10.48550/ARXIV.2512.04795},
  url = {https://arxiv.org/abs/2512.04795},
  author = {Keszegh,  Balázs and Suk,  Andrew and Tardos,  Gábor and Zeng,  Ji},
  title = {Unavoidable patterns and plane paths in dense topological graphs},
  publisher = {arXiv},
  year = {2025},
  copyright = {arXiv.org perpetual,  non-exclusive license}
}

@phdthesis{Alsolami2012Auth,
    title    = {The good drawings $D_n$ of the complete graph $K_n$},
    school   = {McGill University},
    author   = {Rafla, Nabil H.},
    year     = {1988}
}

@article{Faria2008,
  title = {An improved upper bound on the crossing number of the hypercube},
  volume = {59},
  ISSN = {1097-0118},
  url = {http://dx.doi.org/10.1002/jgt.20330},
  DOI = {10.1002/jgt.20330},
  number = {2},
  journal = {Journal of Graph Theory},
  publisher = {Wiley},
  author = {Faria,  Luerbio and Herrera de Figueiredo,  Celina Miraglia and Sýkora,  Ondrej and Vrt’o,  Imrich},
  year = {2008},
  month = jun,
  pages = {145–161}
}

@article{Hammack2023,
 ISSN = {00029890, 19300972},
 URL = {https://www.jstor.org/stable/48662634},
 author = {Richard H. Hammack and Paul C. Kainen},
 journal = {The American Mathematical Monthly},
 number = {4},
 pages = {pp. 352--359},
 publisher = {[Taylor & Francis, Ltd., Mathematical Association of America]},
 title = {A New View of Hypercube Genus},
 urldate = {2026-02-05},
 volume = {128},
 year = {2021}
 
}

@misc{kainen2025,
  doi = {10.48550/ARXIV.2501.02400},
  url = {https://arxiv.org/abs/2501.02400},
  author = {Kainen,  Paul C.},
  title = {Skewness,  crossing number and {E}uler's bound for graphs on surfaces},
  publisher = {arXiv},
  year = {2025},
  copyright = {Creative Commons Attribution Non Commercial Share Alike 4.0 International}
}

@inbook{Shahrokhi1997,
  title = {Bipartite crossing numbers of meshes and hypercubes},
  ISBN = {9783540696742},
  ISSN = {1611-3349},
  url = {http://dx.doi.org/10.1007/3-540-63938-1_48},
  DOI = {10.1007/3-540-63938-1_48},
  booktitle = {Graph Drawing},
  publisher = {Springer Berlin Heidelberg},
  author = {Shahrokhi,  Farhad and Sykora,  Ondrej and Székely,  László A. and Vrt’o,  Imrich},
  year = {1997},
  pages = {37–46}
}

@article{Pach2003,
  title = {Unavoidable Configurations in Complete Topological Graphs},
  volume = {30},
  ISSN = {1432-0444},
  url = {http://dx.doi.org/10.1007/s00454-003-0012-9},
  DOI = {10.1007/s00454-003-0012-9},
  number = {2},
  journal = {Discrete and Computational Geometry},
  publisher = {Springer Science and Business Media LLC},
  author = {Pach,  J\'{a}nos and S\'{o}lymosi,  Jozsef and T\'{o}th,  G\'{e}za},
  year = {2003},
  month = aug,
  pages = {311–320}
}
 
\appendix

\section{Proof of Lemma~\ref{lemma_longplanepaths}}\label{app:A}

\longplanepaths*
\label{lemma_longplanepaths*}

\begin{proof}
    Let $\Gamma$ be a convex-geometric drawing of $G$. We construct a sequence of convex-geometric drawings $\Gamma_0,
\Gamma_1,\Gamma_2,\dots, \Gamma_{k}$ of subgraphs of $G$ in the following way.  We start with  $\Gamma_0 = \Gamma$. Then, we obtain $\Gamma_1$ from $\Gamma_0$ by removing the leftmost edge adjacent to each vertex.  Then, for $i\in\{2,\dots, k\}$ we obtain $\Gamma_i$ from $\Gamma_{i-1}$ by alternatingly removing the rightmost and leftmost edge at every vertex. As every step deletes at most $n$ edges, it is clear by our assumptions that $\Gamma_{k}$ contains at least one edge $e=uv$. In the following we assume that $k$ is odd and therefore $\Gamma_{k}$ was obtained from $\Gamma_{k-1}$ by removing the leftmost edge at every vertex. But since $e$ was not removed from $\Gamma_{k-1}$, it was not the leftmost edge of any of its endpoints. Let $l_u=uu'$ and $l_v=vv'$ be the leftmost edges of $u$ and $v$ respectively in $\Gamma_{k-1}$. Then, $l_u,e,$ and $l_v$ together form a plane path with three edges. Now, since $l_u$ and $l_v$ were not removed in the step from $\Gamma_{k-2}$ to $\Gamma_{k-1}$, it means that $l_u$ and $l_v$ were not the rightmost edges of $u'$ and $v'$ and we can again find two edges in $\Gamma_{k-2}$ that extend this path. Since we can repeat this argument $k$ times, in the end we obtain a path consisting of $2k+1$ edges.
\end{proof}

\end{document}